\documentclass[aps,pra,twocolumn,superscriptaddress,nofootinbib]{revtex4-2}

\usepackage{graphicx}
\usepackage{latexsym}
\usepackage{amsmath}
\usepackage{amssymb}
\usepackage{amsfonts}
\usepackage{color}
\usepackage{pgf}
\usepackage{hyperref}
\usepackage{amsthm}
\usepackage{dsfont}
\usepackage{tikz}
\usetikzlibrary{matrix}

\usepackage[utf8]{inputenc}
\usepackage[english]{babel}
\usepackage{ulem}
\usepackage{mwe} 
\usepackage{mathtools}
\usepackage{amssymb}
\usepackage{colortbl}
\usepackage{braket}
\usepackage{dsfont}

\newtheorem{theorem}{Theorem}
\newtheorem{lemma}{Lemma}
\newtheorem{definition}{Definition}

\begin{document}

\title{Probabilistic pure state conversion on the majorization lattice}

\author{Serge Deside}\affiliation{Centre for Quantum Information and Communication, \'Ecole polytechnique de Bruxelles, CP 165, Universit\'e libre de Bruxelles, 1050 Brussels, Belgium}
\author{Matthieu Arnhem}
\affiliation{Centre for Quantum Information and Communication, \'Ecole polytechnique de Bruxelles, CP 165, Universit\'e libre de Bruxelles, 1050 Brussels, Belgium}
\affiliation{Department of Optics, Palacký University, 17. listopadu 12, 771 46 Olomouc, Czech Republic}
\affiliation{Univ. Lille, CNRS, Inria, UMR 8524 - Laboratoire Paul Painlevé, F-59000 Lille, France}
\author{Célia Griffet}
\affiliation{Centre for Quantum Information and Communication, \'Ecole polytechnique de Bruxelles, CP 165, Universit\'e libre de Bruxelles, 1050 Brussels, Belgium}
\author{Nicolas J. Cerf}
\affiliation{Centre for Quantum Information and Communication, \'Ecole polytechnique de Bruxelles, CP 165, Universit\'e libre de Bruxelles, 1050 Brussels, Belgium}

\begin{abstract}
Entanglement is among the most fundamental—and at the same time puzzling—properties of quantum physics. Its modern description relies on a resource-theoretical approach, which treats entangled systems as a means to enable or accelerate certain informational tasks. Hence, it is of crucial importance to determine whether—and how—different entangled states can be converted into each other under free operations (those which do not create entanglement from nothing). Here, we show that the majorization lattice provides an efficient framework in order to characterize the allowed transformations of pure entangled states under local operations and classical communication. The underlying notions of meet $\land$ and join $\lor$ in the majorization lattice lead us to define, respectively, the optimal common resource and optimal common product states. Based on these two states, we introduce two optimal probabilistic protocols for the (single-copy) conversion of incomparable bipartite pure states, which we name greedy and thrifty. Both protocols reduce to Vidal's protocol [\href{https://doi.org/10.1103/PhysRevLett.83.1046}{G. Vidal, Phys. Rev. Lett. 83, 1046 (1999)}] if the initial and final states are comparable, but otherwise the thrifty protocol can be shown to be superior to the greedy protocol as it yields a more entangled residual state when it fails (they both yield the same entangled state with the same optimal probability when they succeed). Finally, we consider
the generalization of these protocols to entanglement transformations involving multiple initial or final states.
\end{abstract}


\maketitle

\section{Introduction}

Quantum entanglement has long been recognized as a necessary resource in many quantum information protocols \cite{Horodeckis_QuantumEntanglement_2009,BennettWiesner_CommunicationViaOneAndTwoParticleOperators_1992,BennettEtAl_TeleportingAnUnknownQuantumState_1993} and is often regarded as the paramount quantum resource \cite{ChitambarGour_QuantumResourceTheories_2019}. In this context, separable states are defined as resource-free states and local operations supplemented with classical communication (LOCC) are viewed as free operations, that is, operations that map separable states onto separable states (hence, do not create any resource). An important aspect of such a resource theory of entanglement concerns the study of allowed entanglement transformations, \textit{i.e.}, the conversions between entangled states of a system shared between several parties that can be achieved using LOCC only. Note that the conversions of bipartite pure states are asymptotically (in the limit of many copies) always allowed with some yield \cite{BennettBernstein_Concentrating_1996}. Instead, here, we focus on the case of single-copy transformations between bipartite pure states. 
The most notable result on this subject, due to Nielsen \cite{Nielsen_Conditions_1999}, consists in a simple characterization of deterministic bipartite pure state conversions using LOCC by means of majorization theory. A related result, due to Vidal \cite{Vidal_EntanglementOfPureStatesForASingleCopy_1999}, further extends this characterization to probabilistic transformations. 


The theory of majorization allows one to compare probability distributions in terms of intrinsic disorder \cite{MarshallOlkinArnold_Inequalities_2011}. Given two probability vectors $\textbf{\textit{p}}$ and $\textbf{\textit{q}}$, we say that $\textbf{\textit{p}}$ is majorized by $\textbf{\textit{q}}$, written $\textbf{\textit{p}} \prec \textbf{\textit{q}}$, if and only if
\begin{align}
    \label{eq:majo}
    &\sum_{i=1}^k p_i^\downarrow \leq \sum_{i=1}^k q_i^\downarrow, \quad \forall k \in [1, d-1],\\
    \label{equality majorization}
    &\sum_{i=1}^k p_i^\downarrow = \sum_{i=1}^k q_i^\downarrow, \quad \text{for } k = d,
\end{align}
where $p_i^\downarrow$ (resp. $q_i^\downarrow$) denotes the $i^\text{th}$ greatest component of $\textbf{\textit{p}}$ (resp. $\textbf{\textit{q}}$). Here, $d$ is the number of non-zero components of the longest vector between $\textbf{\textit{p}}$ and $\textbf{\textit{q}}$ (zeros are appended at the end of the smallest vector if necessary), hence Eq.~(\ref{equality majorization}) is obviously fulfilled due to  probability normalization. Note that majorization only defines a preorder relation on probability vectors, meaning that two probability vectors may be incomparable under majorization. For example, $(0.5, 0.4, 0.1)$ and $(0.6, 0.2, 0.2)$ do not fulfill Eq.~(\ref{eq:majo}) in either direction. On the opposite, if Eq.~(\ref{eq:majo}) holds in both directions ($\textbf{\textit{p}} \prec \textbf{\textit{q}}$ and $\textbf{\textit{q}} \prec \textbf{\textit{p}}$), then the two vectors are said to be equivalent, which means that they coincide up to a permutation.

Majorization relations have a central role in entanglement theory, as implied by Nielsen's theorem \cite{Nielsen_Conditions_1999}. Consider two parties, Alice and Bob, sharing a bipartite pure state $\ket{\psi}$. The Schmidt decomposition of $\ket{\psi}$ is written
\begin{equation}
    \ket{\psi}_{AB} = \sum_{i=1}^d \sqrt{\lambda_\psi^{(i)}} \ket{i}_A \ket{i}_B,
\end{equation} 
where $\lambda_\psi = \left(\lambda_\psi^{(1)}, \cdots, \lambda_\psi^{(d)}\right) $ is the Schmidt vector of $\ket{\psi}$ made with the eigenvalues of the reduced density matrix $\hat{\rho}_A$ (or $\hat{\rho}_B$), while $\{\ket{i}_A\}$ (resp. $\{\ket{i}_B\}$) denotes the eigenbasis of $\hat{\rho}_A$ (resp. $\hat{\rho}_B$). If Alice and Bob wish to convert $\ket{\psi}_{AB}$ into another pure state 
\begin{equation}
\ket{\phi}_{AB} = \sum_{i=1}^d \sqrt{\lambda_\phi^{(i)}} \ket{i}_A \ket{i}_B,
\end{equation}
this entanglement transformation is possible with certainty using LOCC if and only if 
\begin{equation}
\label{eq:nielsen majorization}
    \lambda_\psi \prec \lambda_\phi.
\end{equation}
Of course, two states are locally unitarily equivalent, \textit{i.e.}, interconvertible under local unitaries $\hat U_A\otimes \hat U_B$, if and only if they share the same Schmidt vector up to a permutation. Hence, it is sufficient to consider Schmidt vectors sorted by decreasing order in order to compare them with a majorization relation (in what follows, all Schmidt vectors will always be assumed to be sorted decreasingly, that is $\lambda_\psi^{(i)\downarrow}=\lambda_\psi^{(i)}$, so we will omit the arrow sign).

An equivalent form of Nielsen's theorem can be stated by using a sufficiently large set of entanglement monotones \cite{Vidal_EntanglementMonotones_2000}, \textit{i.e.}, functionals of the state $\ket{\psi}$ that cannot increase, on average, under a LOCC transformation. By using Eqs. \eqref{eq:majo} and \eqref{equality majorization}, one can easily check that
\begin{equation}
\label{eq:definition of entanglement monotones}
    E_l(\psi) = \sum_{i = l}^d \lambda_\psi^{(i)}, \quad \forall l \in [1,d],
\end{equation}
form a set of $d$ entanglement monotones which allow us to reexpress Nielsen's theorem as
\begin{equation}
\ket\psi \overset{\text{LOCC}}\longrightarrow \ket\phi ~~ \Longleftrightarrow ~~ E_l(\psi) \geq E_l(\phi),\quad \forall l \in [1,d],
\end{equation}
where $\ket\psi \overset{\text{LOCC}}\longrightarrow \ket\phi$ means that $\ket\psi$ can be deterministically converted into $\ket\phi$ by using LOCC. Protocols achieving such a transformation have been discussed, for example, in Refs. \cite{Jensen_SimpleAlgorithmForLocalConversion_2001, TorunYildiz_DeterministicTransformations_2015}.
The characterization of entanglement transformations via majorization relations has been generalized to probabilistic LOCC transformations by Vidal in Ref. \cite{Vidal_EntanglementOfPureStatesForASingleCopy_1999}, and an optimal protocol achieving the desired probabilistic transformation was also provided (it will be detailed in Sec. \ref{sec:probabilistic state conversions}). Optimality refers here to a protocol that reaches the exact state $\ket\phi$ with the highest possible success probability.

In this paper, we exploit the  \textit{majorization lattice} \cite{CicaleseVaccaro_Supermodularity_2002}, a notion recently shown to be relevant for addressing quantum information questions \cite{Korzekwa_StructureOfTheThermodynamicArrow_2017, BosykEtAl_ApproximateTransformationsOfBipartite_2017, BosykFreytesBellomo_LatticeOfTrumpingMajorization_2018, BosykBellomoHolikFreytes_OptimalCommonResource_2019, MassriEtAl_ExtremalElementsOfASublattice_2020, SauerweinSchwaigerKraus_DiscreteAndDifferentiableEntanglementTransformations_2018, deOliveiraJuniorEtAl_GeometricStructure_2022, YuGühne_DetectingCoherence_2019}, in order to revisit the probabilistic entanglement transformations between bipartite pure states. Our 
analysis builds upon two central elements of a lattice theory, the so-called \textit{meet} and \textit{join} (see Fig. \ref{fig:majorization lattice}), and yields two corresponding optimal protocols (see Fig. \ref{fig:majorization lattice for bipartite pure states}). Considering an entanglement lattice, where each node is a bipartite pure state and the lattice structure emerges from majorization relations, we associate the meet and join to two specific states that we call, respectively, the optimal common resource (OCR)\footnote{\label{note1}This term was first coined in Ref. \cite{GuoChitambarDuan_CommonResourceState_2016} but with no reference to the majorization lattice.} and optimal common product (OCP) states. We first build a protocol using the OCP state and show that it is akin to Vidal's optimal protocol. More interestingly, we then build a second protocol making use of the OCR state and prove that, while being again optimal, it better preserves average entanglement. Both protocols trivially reduce to Vidal's in the special case where $\ket\psi$ and $\ket\phi$ are comparable states in the sense of majorization theory, so we focus on the interesting case of incomparable states.

The protocol based on  the OCP state can be viewed as a \textit{greedy} protocol. In the context of optimization, a greedy algorithm evolves towards the solution by choosing the local optimal move at each stage. In analogy, our greedy protocol favors the immediate gain of a deterministic move and postpones the probabilistic move (see red arrows in Fig.~\ref{fig:majorization lattice for bipartite pure states}). Our second protocol based on the OCR state can be understood, in contrast, as a \textit{thrifty} protocol. A thrifty optimization algorithm prefers not to choose immediate gain at each stage. Here, our thrifty protocol indeed tolerates starting with a probabilistic move (see green arrows in Fig.  \ref{fig:majorization lattice for bipartite pure states}). The main results of this paper consist in Theorems \ref{thm:optimal proba} and \ref{thm:better entanglement resource}, which prove, respectively, the optimality and better average entanglement preservation of the thrifty protocol.



The paper is organized as follows. In Sec.~\ref{sec:majorization lattice}, we first introduce the majorization lattice and, specifically for bipartite entanglement, the OCR and OCP states. Then, in Sec.~\ref{sec:probabilistic state conversions}, we present Vidal's theorem for the probabilistic conversion of bipartite pure states  as well as the corresponding optimal protocol. This provides us with the tools for introducing the greedy and thrifty protocols in Sec. \ref{sec:new protocol}. We further establish in Sec. \ref{sec:generalization to multiple final states} the generalization of the greedy and thrifty protocols to an arbitrary number of initial or final states. Finally, we give our conclusions in Sec. \ref{sec:conclusion}.


\section{Majorization lattice}
\label{sec:majorization lattice}


Arising from order theory, the notion of lattice unveils a partial order relation on the elements of a set. Hereafter, we consider a specific lattice where the elements under comparison belong to the set of probability distributions sorted decreasingly and where the partial order is given by the majorization relation. More formally, the majorization lattice \cite{CicaleseVaccaro_Supermodularity_2002} is a quadruple $\langle \mathcal{P}_d, \prec, \land, \lor \rangle$, where
\begin{itemize}
    \item $\mathcal{P}_d$ is the set of discrete probability vectors sorted decreasingly with, at most, $d$ non-zero coefficients, that is, $\mathcal{P}_d = \{ (p_1, \cdots, p_d) \text{, s.t. } p_1\geq \cdots \geq p_d \geq 0 \text{ and } \sum_{i=1}^d p_i =  1\}$;
    \item $\prec$ is the majorization relation;  
    \item $\land$ denotes the so-called \textit{meet};
    \item $\lor$ denotes the so-called \textit{join}.
\end{itemize}

The meet of two elements $\textbf{\textit{p}}, \textbf{\textit{q}} \in \mathcal{P}_d$, denoted $\textbf{\textit{p}} \land \textbf{\textit{q}}$, is defined as the sole element of $\mathcal{P}_d$ such that, $\forall \textbf{\textit{r}} \in \mathcal{P}_d$ with $\textbf{\textit{r}} \prec \textbf{\textit{p}}$ and $\textbf{\textit{r}} \prec \textbf{\textit{q}}$, we have $\textbf{\textit{r}} \prec \textbf{\textit{p}} \land \textbf{\textit{q}}$. Analogously, the join of two elements $\textbf{\textit{p}}, \textbf{\textit{q}} \in \mathcal{P}_d$, denoted $\textbf{\textit{p}} \lor \textbf{\textit{q}}$, is defined as the sole element of $\mathcal{P}_d$ such that, $\forall \textbf{\textit{r}} \in \mathcal{P}_d$ with $\textbf{\textit{p}} \prec \textbf{\textit{r}}$ and $\textbf{\textit{q}} \prec \textbf{\textit{r}}$, we have $\textbf{\textit{p}} \lor \textbf{\textit{q}} \prec \textbf{\textit{r}}$. 

The following definition allows one to characterize mathematically the meet of two vectors in $\mathcal{P}_d$.

\begin{definition}\textnormal{\cite{CicaleseVaccaro_Supermodularity_2002}}
\label{def:meet}
    Let \textbf{p} and \textbf{q} $\in \mathcal{P}_d$, the meet of \textbf{p} and \textbf{q}, denoted $\textbf{p} \land \textbf{q} = (m_1, m_2, \cdots, m_d)$, can be expressed as
    \begin{equation}
        m_i = \min \left\{\sum_{j=1}^i p_j, \sum_{j=1}^i q_j \right\} - \min \left\{\sum_{j=1}^{i-1} p_j, \sum_{j=1}^{i-1} q_j \right\},
    \end{equation}
     for all $i \in [1,d]$, with the convention that $\sum_{j=1}^k p_j = \sum_{j=1}^k q_j = 0$ when $k=0$.
\end{definition}

The determination of the join of two vectors in $\mathcal{P}_d$ is more cumbersome. As we will not use it explicitly in the following, we refer the interested reader to Ref. \cite{CicaleseVaccaro_Supermodularity_2002} for an algorithm producing the join.


Figure \ref{fig:majorization lattice} provides a useful visual insight into the majorization lattice where, for each element $\textbf{\textit{x}}$ of the lattice, an upper and a lower cone are defined. The elements $\textbf{\textit{x}}_\text{up}$ of the upper cone correspond to all elements that are majorized by $\textbf{\textit{x}}$, \textit{i.e.}, $\textbf{\textit{x}}_\text{up} \prec \textbf{\textit{x}}$, while the elements $\textbf{\textit{x}}_\text{down}$ of the lower cone correspond to all elements majorizing $\textbf{\textit{x}}$, \textit{i.e.}, $\textbf{\textit{x}} \prec \textbf{\textit{x}}_\text{down}$. In other words, the upper cone comprises all the elements that are more disordered than $\textbf{\textit{x}}$, whereas the lower cone comprises all the elements that are more ordered than $\textbf{\textit{x}}$. Considering now two elements \textbf{\textit{p}} and \textbf{\textit{q}} in the lattice, the meet $\textbf{\textit{p}} \land \textbf{\textit{q}}$ can be viewed as the most ordered element that remains more disordered than both \textbf{\textit{p}} and \textbf{\textit{q}}. Conversely, the join $\textbf{\textit{p}} \lor \textbf{\textit{q}}$ corresponds to the most disordered element that remains more ordered than both \textbf{\textit{p}} and \textbf{\textit{q}}. It is worth mentioning that the meet and join are useful notions only when the two probability vectors \textbf{\textit{p}} and \textbf{\textit{q}} are incomparable under majorization (otherwise, if 
$\textbf{\textit{p}} \prec \textbf{\textit{q}}$, then $\textbf{\textit{p}} \land \textbf{\textit{q}} = \textbf{\textit{p}} $ and $\textbf{\textit{p}} \lor \textbf{\textit{q}} = \textbf{\textit{q}}$).
We will thus be mostly concerned with incomparable Schmidt vectors when developing the greedy and thrifty protocols in Sec.~\ref{sec:new protocol}.

\begin{figure}
    \centering    
    \includegraphics[scale=0.9]{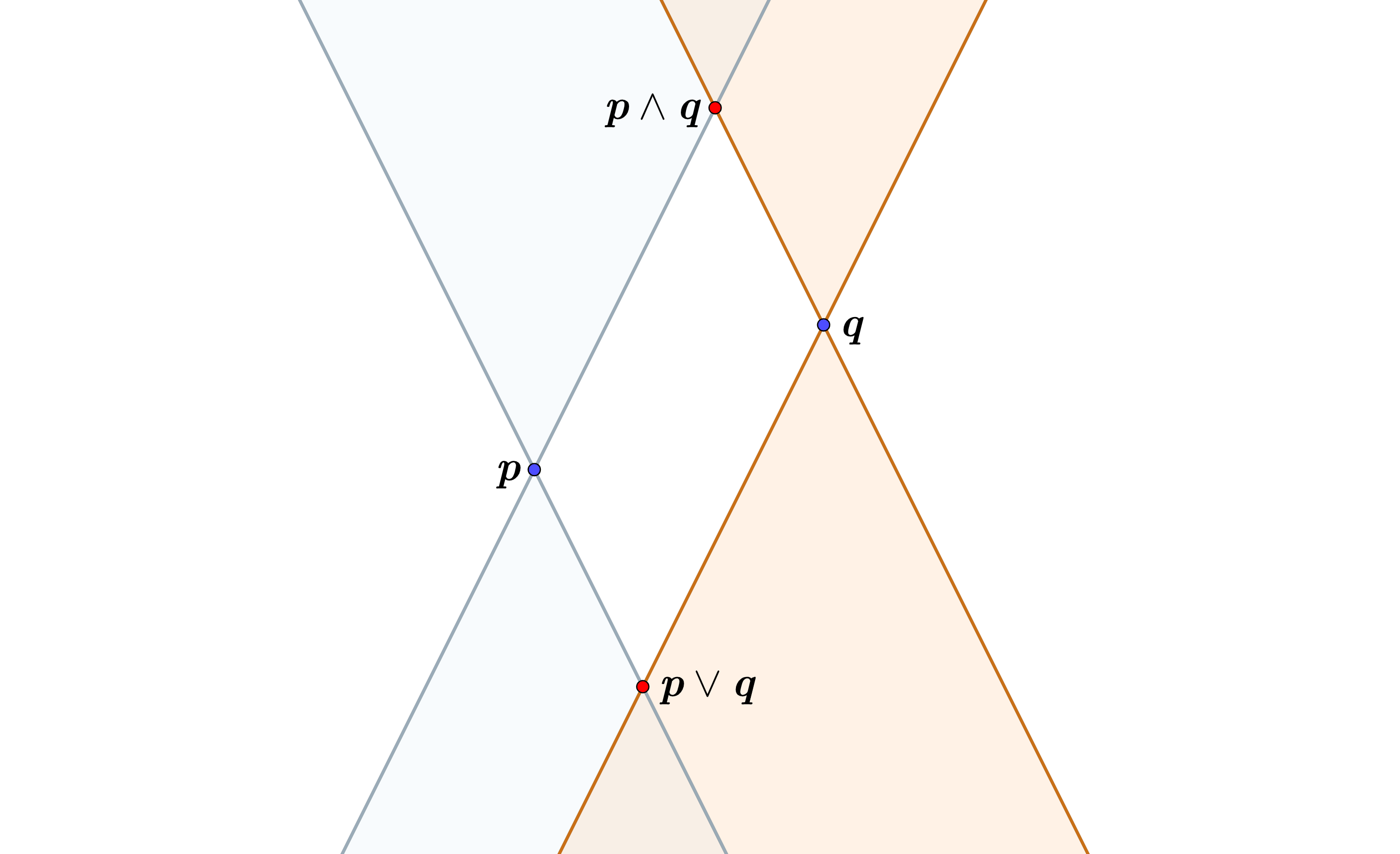}
    \caption{Schematic representation of the majorization lattice. Let $\textbf{\textit{p}}$ and $\textbf{\textit{q}} \in \mathcal{P}_d$, the meet $\textbf{\textit{p}} \land \textbf{\textit{q}}$ is located at the intersection of the two upper cones, whereas the join $\textbf{\textit{p}} \lor \textbf{\textit{q}}$ is located at the intersection of the two lower cones. Since the upper cones contain more disordered elements, the meet $\textbf{\textit{p}} \land \textbf{\textit{q}}$  stands for the most ordered element among those that are more disordered than both $\textbf{\textit{p}}$ and $\textbf{\textit{q}}$. Conversely, since the lower cones contain more ordered elements, the join $\textbf{\textit{p}} \lor \textbf{\textit{q}}$ stands for the most disordered element among those that are more ordered than both $\textbf{\textit{p}}$ and $\textbf{\textit{q}}$.}
    \label{fig:majorization lattice}
\end{figure}


In the following, we will consider an entanglement lattice whose elements are bipartite pure states (Fig. \ref{fig:majorization lattice for bipartite pure states}). Since the entanglement of a bipartite pure state is univocally characterized by its Schmidt vector, we may equivalently view the elements of the lattice as probability vectors in $\mathcal{P}_d$. Hence, according to Nielsen's theorem, a state can be obtained from any state belonging to its upper cone by using only deterministic LOCC. Conversely, a state can be converted into any state of its lower cone by using only deterministic LOCC. In other words, the lattice structure translates the allowed entanglement transformations. Accordingly, the meet of two (or more) states denotes the least entangled state that can still be deterministically converted into any one of them. We refer to it as the optimal common resource (OCR)\footref{note1} state. Reciprocally, the join of two (or more) states denotes the most entangled state that can still be deterministically produced from any one of them. We refer to it as the optimal common product (OCP) state. 
The OCP and OCR states will be used in Sec.~\ref{sec:new protocol} in order to construct the greedy and thrifty protocols, respectively, as pictured in Fig. \ref{fig:majorization lattice for bipartite pure states}.

Instructively, the optimality of the OCP and OCR states can be rephrased as follows. Considering two bipartite pure states $\ket{\psi}$ and $\ket{\phi}$, any state $\ket{\tau}$ that is majorized by both $\ket{\psi}$ and $\ket{\phi}$ (\textit{i.e.}, any state $\ket{\tau}$ lying inside the intersection of the two upper cones emerging from $\ket{\psi}$ and $\ket{\phi}$ in Fig. \ref{fig:majorization lattice for bipartite pure states}) must also be majorized by the OCR state $\ket{\psi \land \phi}$. Thus, according to Nielsen's theorem, all ``resource states" (\textit{i.e.}, states $\ket{\tau}$ that are convertible into either $\ket{\psi}$ or $\ket{\phi}$ via a deterministic LOCC) can also produce the OCR state $\ket{\psi \land \phi}$. Since the latter is itself convertible into either $\ket{\psi}$ or $\ket{\phi}$, it is the resource state that requires the lowest amount of resource. Conversely,  any state $\ket{\sigma}$
that majorizes both $\ket{\psi}$ and $\ket{\phi}$ (\textit{i.e.}, any state $\ket{\sigma}$ lying inside the intersection of the two lower cones emerging from $\ket{\psi}$ and $\ket{\phi}$ in Fig. \ref{fig:majorization lattice for bipartite pure states}) must necessarily also majorize the OCP state $\ket{\psi \lor \phi}$. This means that all ``producible states" (\textit{i.e.}, states $\ket{\sigma}$ that can be deterministically produced with LOCC either from $\ket{\psi}$ or from $\ket{\phi}$) can also be produced from the OCP state $\ket{\psi \lor \phi}$. Since the latter is itself obtainable either from $\ket{\psi}$ or from $\ket{\phi}$, it is the producible state that contains the highest amount of resource.


\begin{figure}[t!]
    \centering    
    \includegraphics[scale=0.9]{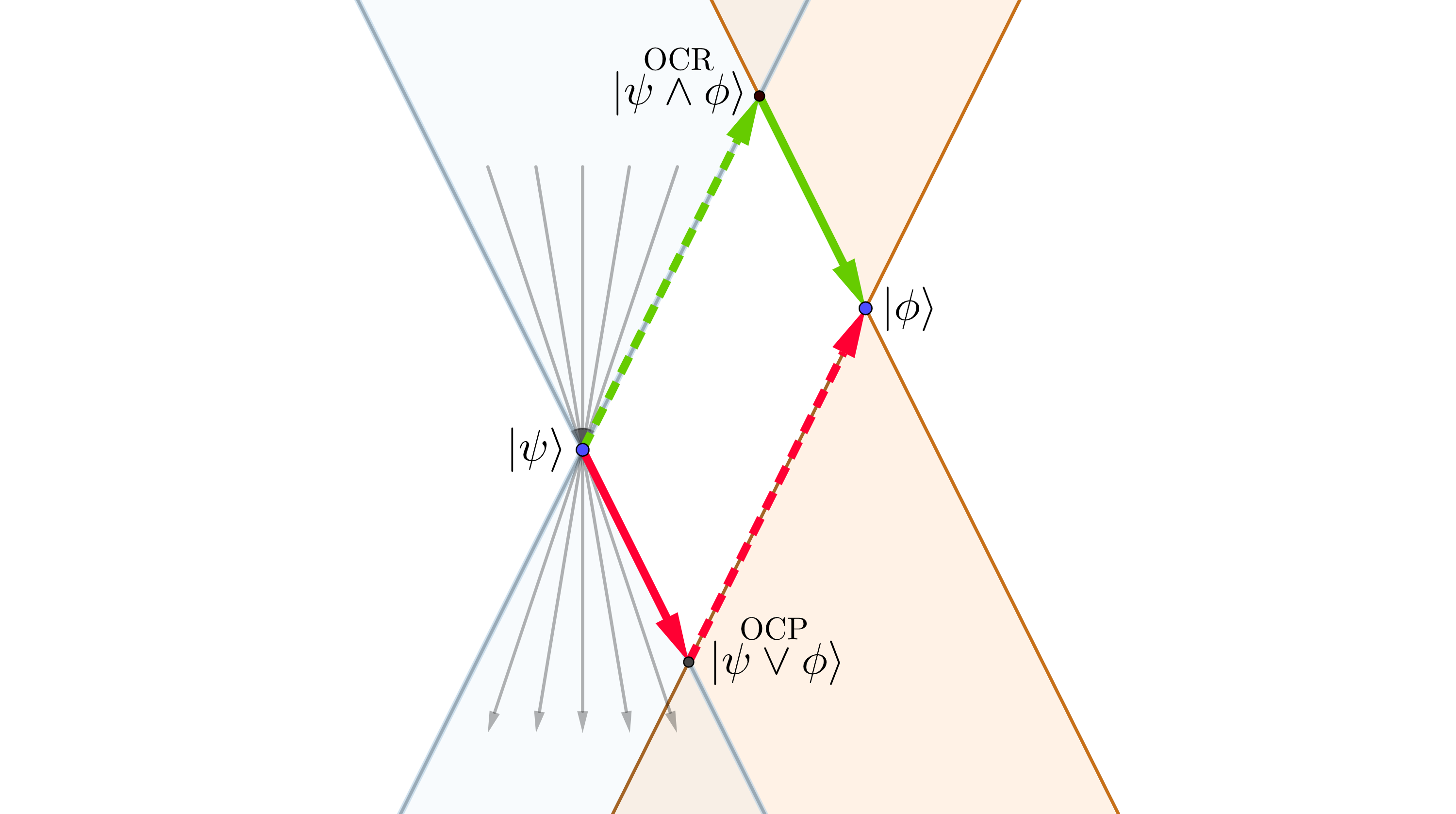}
    \caption{Schematic representation of the majorization lattice applied to bipartite entanglement. Let $\ket\psi$ and $\ket\phi$ be two bipartite pure states that are assumed to be incomparable. Any state lying in the upper cone of $\ket\psi$ is convertible into $\ket\psi$ via deterministic LOCC, while any state in the lower cone of $\ket\psi$ is reachable from $\ket\psi$ via deterministic LOCC (this is pictured with a collection of grey arrows pointing downwards). The same holds true, of course, for the cones associated with $\ket\phi$. The OCR state $\ket{\psi \land \phi}$ is located at the intersection of the two upper cones and can be understood as the least entangled state capable of deterministically producing either $\ket\psi$ or $\ket\phi$. In contrast, the OCP state $\ket{\psi \lor \phi}$ is located at the intersection of the two lower cones and can be viewed as the most entangled state that both $\ket\psi$ and $\ket\phi$ are capable to produce deterministically. We highlight the greedy (red) and thrifty (green) protocols as discussed in Sec.~\ref{sec:new protocol} for converting $\ket\psi$ into $\ket\phi$ (bold arrows correspond to deterministic LOCC transformations and dashed arrows to probabilistic LOCC transformations). The greedy protocol (in red) starts with a deterministic LOCC towards the OCP state and postpones the probabilistic LOCC. In contrast, the thrifty protocol (in green) starts with a probabilistic LOCC towards the OCR state before making the deterministic LOCC.}
    \label{fig:majorization lattice for bipartite pure states}
\end{figure}

\section{Probabilistic state conversion protocol}
\label{sec:probabilistic state conversions}

When a deterministic state conversion is impossible according to Nielsen's theorem, the same transformation may sometimes be achieved probabilistically. This is the content of Vidal's theorem \cite{Vidal_EntanglementOfPureStatesForASingleCopy_1999}, which we now present.

Let $\ket{\psi}$ be the initial state of a bipartite system shared by Alice and Bob, and let $\ket{\phi}$ be the target state they wish to obtain using only LOCC. We suppose that it is possible to perform the desired transformation from $\ket{\psi}$ to $\ket{\phi}$ with non-vanishing probability, \textit{i.e.}, the number of non-zero coefficients of $\lambda_{\psi}$ must be at least as large as that of $\lambda_{\phi}$ (in the following, we consider, without loss of generality, that they are equal). Vidal's theorem states that there exists a probabilistic protocol converting $\ket\psi$ into $\ket\phi$ with probability $p$ if and only if $E_l(\psi) \geq p \, E_l(\phi), \forall l \in [1,d]$. Thus, the optimal (maximum) probability with which Alice and Bob can perform the desired transformation is given by 
\begin{equation}
\label{eq:optimal proba conversion}
    p_{\max} = \underset{l \in [1, d]}{\min}\frac{E_{l}(\psi)}{E_{l}(\phi)}.
\end{equation}
Note that $p_{\max}\le 1$ because when $l=1$, we have $E_{1}(\psi)=E_{1}(\phi)=1$. 
Vidal described a protocol achieving this optimal probability in Ref.  \cite{Vidal_EntanglementOfPureStatesForASingleCopy_1999}, which we recall hereafter.

First, Alice and Bob apply  a deterministic LOCC in order to transform $\ket{\psi}$ into an intermediate state $\ket{\chi}$ which maximizes the fidelity with respect to the target state $\ket\phi$ while being deterministically reachable from $\ket{\psi}$ using only LOCC \cite{VidalJonathanNielsen_ApproximateTransformations_2000}, namely,
\begin{equation}
\label{eq:max fidelity state}
    \ket{\chi} = \underset{\ket{\alpha} ~:~ \ket{\psi} \overset{\text{LOCC}}\longrightarrow \ket{\alpha}}{\text{argmax}} |\langle\alpha|\phi\rangle|^2.
\end{equation}

Second, Alice performs a two-outcome measurement on her share of the system, leading, with some probability $p_{\max}$, to the target state $\ket{\phi}$ (in which case we say that the protocol has succeeded) or, with probability $1-p_{\max}$, to a state denoted as $\ket{\xi}$ (in which case we say that the protocol has failed). We show below that this residual state $\ket{\xi}$ always possesses a strictly smaller number of non-vanishing Schmidt coefficients than $\ket{\phi}$, hence the conversion from $\ket{\xi}$ to $\ket{\phi}$ is fully impossible, even probabilistically.
Note that Vidal's conversion protocol, which we refer simply to as a probabilistic LOCC transformation in the following, involves both a deterministic step and a probabilistic step.

The construction of the intermediate state $\ket{\chi}$ works as follows. We make use of the $d$ entanglement monotones 
defined in Eq. (\ref{eq:definition of entanglement monotones}) in order to define a sequence of ratios $r_j$'s. First, we define $r_1$, which corresponds to the maximum conversion probability from $\ket{\psi}$ to $\ket{\phi}$, as
\begin{equation}
\label{eq:r1}
    r_1 = \underset{l \in [1, d]}{\min}\frac{E_{l}(\psi)}{E_{l}(\phi)} \equiv \frac{E_{l_1}(\psi)}{E_{l_1}(\phi)} ,
\end{equation}
where $l_1$ is the value of $l$ that reaches the minimum (it is
chosen as the smallest value of $l$ in case of several minima).
Then, we define the next ratios $r_j$'s with $j=2,3,\cdots$ as follows (each $r_j$ is associated with a corresponding $l_j$) 
\begin{equation}
\label{eq:ri}
    r_j = \underset{l \in [1, l_{j-1}-1]}{\min} \frac{E_l(\psi)-E_{l_{j-1}}(\psi)}{E_l(\phi)-E_{l_{j-1}}(\phi)} \equiv \frac{E_{l_j}(\psi)-E_{l_{j-1}}(\psi)}{E_{l_j}(\phi)-E_{l_{j-1}}(\phi)} ,
\end{equation}
until we find some value of $j$, which we call $k$, satisfying $l_k = 1$. Finally, we define $l_0 = d+1$. On a side note, it can easily be shown that the sequences of $r_j$'s and $l_j$'s satisfy $0 < r_1 < \cdots < r_k$ and $l_0 > l_1 > \cdots > l_k=1$  \cite{Vidal_EntanglementOfPureStatesForASingleCopy_1999}.

From the ratios $r_j$'s, we define the Schmidt coefficients of the intermediate state
\begin{equation}
\label{eq:approximate state Vidal}
    \lambda_\chi^{(i)} = r_j\lambda_{\phi}^{(i)}, \quad \text{if~}  i \in [l_j, l_{j-1}-1],\quad j\in[1,k],
\end{equation}
which, by construction, satisfies $\lambda_\psi \prec \lambda_\chi$. Hence, $\ket\psi$ is convertible to $\ket\chi$ by using a deterministic LOCC.

The generalized measurement performed by Alice is described by the two Kraus operators 
\begin{gather}
    \hat M = \begin{pmatrix}
        \hat M_k & & \\
         & \ddots & \\
         & & \hat M_1
    \end{pmatrix},\\
    \label{N}
    \hat N = \begin{pmatrix}
         \sqrt{\hat I_{[l_{k-1}-l_k]}- \hat{M}_k^2} & & \\
         & \ddots & \\
         & & \sqrt{\hat I_{[l_0-l_1]}-\hat{M}_1^2}
    \end{pmatrix},
\end{gather}
with 
\begin{equation}
    \hat M_j = \sqrt{\frac{r_1}{r_j}} \,\hat I_{[l_{j-1}-l_j]}, \quad \text{for~}  j \in [1, k],
\end{equation}
where $\hat I_{[l_{j-1}-l_j]}$ is the identity operator in a $(l_{j-1}-l_j)$-dimensional Hilbert space. We have, of course, the completeness relation $\hat M^\dagger\hat M+\hat N^\dagger\hat N= \hat I_{[d]}$.
If Alice obtains the measurement outcome linked to $\hat M$, she obtains the target state $\ket{\phi}$. This happens with probability $r_1$ since we have
\begin{equation}
    \hat M \otimes \hat I \ket{\chi} = \sqrt{r_1} \ket{\phi}.
\end{equation}
Otherwise, if she obtains the outcome linked to $\hat N$, she gets the residual state $\ket{\xi}$ such that  
\begin{equation}
    \hat N \otimes \hat I \ket{\chi} = \sqrt{1-r_1} \ket{\xi}.
\end{equation}

It is clear from Eq. \eqref{N} that $\hat N$ possesses at least one vanishing diagonal element, hence $\ket{\xi}$ has strictly less non-vanishing Schmidt coefficients than $\ket{\phi}$. Therefore, it is impossible for Alice and Bob to obtain state $\ket{\phi}$ from $\ket{\xi}$ by using (even probabilistic) LOCC transformations.

Note that in the special case where $\lambda_{\psi} \prec \lambda_{\phi}$, Vidal's protocol reduces to a deterministic  ($p_{\max}=1$)
 protocol for converting $\ket{\psi}$ into $\ket{\phi}$, which exists as a consequence of Nielsen's theorem.

\section{Greedy and thrifty protocols on the majorization lattice}
\label{sec:new protocol}


We now introduce two optimal probabilistic state conversion protocols which take advantage of the notions of OCP and OCR states on the majorization lattice (see Fig. \ref{fig:lattice_two_protocols} for a precise diagrammatic representation of both protocols). As usual, Alice and Bob initially share a pure bipartite state $\ket{\psi}$ that they wish to transform into a target state $\ket\phi$ by using only LOCC operations. Further, we assume that $\ket\psi$ and $\ket\phi$ are incomparable. The protocol passing through the OCP state works as follows.  First, Alice and Bob transform $\ket{\psi}$ into the OCP of the initial and target states, denoted as $\ket{\psi \lor \phi}$, via a deterministic LOCC transformation. This conversion is possible as a consequence of Nielsen's theorem since, by definition, $\lambda_\psi \prec \lambda_{\psi \lor \phi}$. Then, Alice and Bob apply to $\ket{\psi \lor \phi}$ the probabilistic LOCC transformation\footnote{It consists of a deterministic step from $\ket{\psi \lor \phi}$ to the intermediate state $\ket{\chi}$, followed by a probabilistic step from $\ket{\chi}$ to $\ket\phi$.} presented in Sec. \ref{sec:probabilistic state conversions}  in order to obtain the target state $\ket\phi$. This clearly cannot be done deterministically since $\lambda_{\psi \lor \phi} \nprec \lambda_{\phi}$ (actually, $\lambda_{\psi \lor \phi} \succ \lambda_{\phi}$). We name this protocol \textit{greedy} because the  deterministic (rewarding) LOCC transformation is given priority and precedes the probabilistic (less rewarding) LOCC transformation.

Following a similar logic, we define another protocol using the OCR state instead. In this protocol, Alice and Bob first transform $\ket{\psi}$ into the OCR of the initial and target states, denoted as $\ket{\psi \land \phi}$, using the probabilistic LOCC transformation\footnote{It consists of a deterministic step from $\ket{\psi}$ to the intermediate state $\ket{\zeta}$, followed by a probabilistic step from $\ket{\zeta}$ to $\ket{\psi \land \phi}$.} presented in Sec. \ref{sec:probabilistic state conversions}. This transformation can only be made probabilistically  because $\lambda_\psi  \nprec\lambda_{\psi \land \phi}$ (actually, $\lambda_\psi \succ \lambda_{\psi \land \phi}$). Then, the second phase consists in Alice and Bob transforming $\ket{\psi \land \phi}$ into $\ket{\phi}$ via a deterministic LOCC operation, which is obviously possible because $\lambda_{\psi \land \phi} \prec \lambda_\phi$. We name this protocol \textit{thrifty} because the deterministic (rewarding) LOCC transformation is postponed after the probabilistic (less rewarding) LOCC transformation. 


Interestingly, we will prove that the greedy and thrifty protocols both reach the optimal conversion probability [the same as that of Vidal's protocol, see Eq. \eqref{eq:optimal proba conversion}]. Furthermore, we will show that the thrifty protocol better preserves entanglement, on average, than the greedy protocol. More precisely, if the thrifty protocol fails, the residual state, which we call $\ket{\nu}$, is more entangled than the residual state of the greedy protocol, called $\ket{\xi}$, in the sense that $\lambda_\nu \prec \lambda_\xi$, or equivalently $\ket\nu \overset{\text{LOCC}}\longrightarrow \ket\xi$.

\begin{figure}
    \centering    
    \includegraphics[scale=0.75]{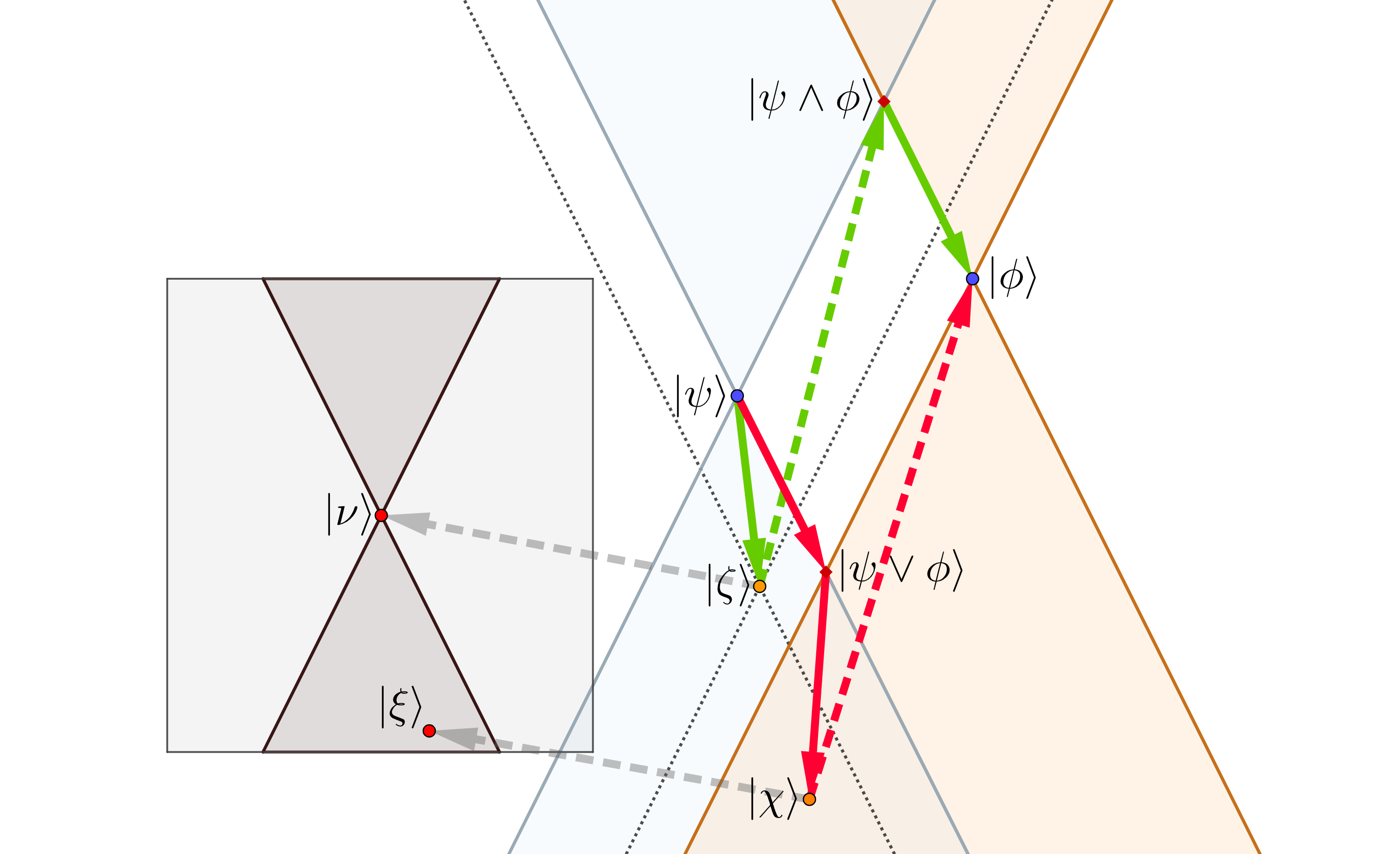}
    \caption{Diagrammatic representation of the greedy (red) and thrifty (green) protocols converting $\ket\psi$ into $\ket\phi$ on the majorization lattice. The probabilistic LOCC transformation of both protocols is itself split into two steps: a deterministic step (represented by a bold arrow) towards an intermediate state (either $\ket\chi$ or $\ket\zeta$) followed by a purely probabilistic step (represented by a dashed arrow) involving a local two-outcome measurement. If the conversion is successful, the resulting state is majorized by (is more entangled than) the initial state  ($\ket \phi \prec \ket{\psi \lor \phi}$ and $\ket{\psi \land \phi} \prec \ket \psi$). The left inset represents the residual states if the conversion has failed, emphasizing that the residual state $\ket\nu$ of the thrifty protocol is majorized by (hence, is more entangled than) the residual state $\ket\xi$ of the greedy protocol. Note also that $\ket{\zeta} \prec \ket{\chi}$.
       }
    \label{fig:lattice_two_protocols}
\end{figure}





Note that if $\ket\psi$ and $\ket\phi$ were comparable instead of incomparable states, then the greedy and thrifty protocols would be equal and would coincide with Vidal's conversion protocol from $\ket\psi$ to $\ket\phi$ (which is probabilistic when $\lambda_\phi \prec \lambda_\psi$, or even deterministic when $\lambda_\psi \prec \lambda_\phi$). Hence, we will only be interested in the case of incomparable Schmidt vectors $\lambda_\psi$ and $\lambda_\phi$, \textit{i.e.}, situations where a deterministic LOCC conversion is impossible in both directions, from $\ket\psi$ to $\ket\phi$ and from $\ket\phi$ to $\ket\psi$.


Let us first consider the easiest case of the greedy protocol, which can be brought to Vidal's protocol simply by merging the sequence of two deterministic steps $\ket\psi \to \ket{\psi \lor \phi} \to \ket\chi$ into a single deterministic step $\ket\psi \to \ket\chi$. Indeed, in the second phase of the greedy protocol (\textit{i.e.}, the probabilistic LOCC transformation from $\ket{\psi \lor \phi}$ to $\ket\phi$), the intermediate state happens to be the same as the intermediate state $\ket\chi$ of the transformation from $\ket\psi$ to $\ket\phi$ in Vidal's protocol, defined in Eq. (\ref{eq:approximate state Vidal}). This is true because this intermediate state $\ket\chi$ has been shown to belong to the lower cone of $\ket{\psi \lor \phi}$ in Ref.  \cite{BosykEtAl_ApproximateTransformationsOfBipartite_2017}, and is thus deterministically reachable from $\ket\psi$. Since, by definition, the intermediate state is the closest state to $\ket\phi$ that is  deterministically attainable from $\ket\psi$ [see Eq. (\ref{eq:max fidelity state})], both intermediate states coincide. Hence, the greedy protocol is also optimal.

The rest of this section is devoted to the analysis of the thrifty protocol, which is more involved. Before proving the optimality of the thrifty protocol, let us prove the following rather intuitive lemma.

\begin{lemma}
\label{lemma:entanglement monotones}
    The entanglement monotone for the OCR of two bipartite pure states $\ket{\psi}$ and $\ket{\phi}$, $E_l(\psi \land \phi)$, is equal to the maximum between the entanglement monotones of the initial state $E_l(\psi)$ and  target state $E_l(\phi)$, \textit{i.e.},
    \begin{equation}
        E_l(\psi \land \phi) = \max\left\{{E_l(\psi), E_l(\phi)}\right\}, \quad \forall l \in [1,d].
    \end{equation}
\end{lemma}

\begin{proof}
    We start from the definition of the $l$th entanglement monotone for the OCR state,
    \begin{equation}
        E_l(\psi \land \phi) = \sum_{i=l}^d \lambda_{\psi \land \phi}^{(i)},
    \end{equation}
    which can be rewritten, by means of Definition \ref{def:meet}, as
    \begin{multline}
        E_l(\psi \land \phi) =  \sum_{i=l}^d \left( \min\left\{\sum_{j=1}^i \lambda_\psi^{(j)}, \sum_{j=1}^i \lambda_\phi^{(j)}\right\} \right.\\
        - \left. \min\left\{\sum_{j=1}^{i-1} \lambda_\psi^{(j)}, \sum_{j=1}^{i-1} \lambda_\phi^{(j)}\right\}\right).
    \end{multline}
    Since all terms cancel out two by two, except for the first and last terms, we obtain
    \begin{equation}
        E_l(\psi \land \phi) = 1- \min\left\{\sum_{j=1}^{l-1}\lambda_\psi^{(j)}, \sum_{j=1}^{l-1}\lambda_\phi^{(j)}\right\},
    \end{equation}
    or equivalently,
    \begin{align}
        E_l(\psi \land \phi) &= \max\left\{\sum_{j=l}^{d}\lambda_\psi^{(j)}, \sum_{j=l}^{d}\lambda_\phi^{(j)}\right\},\\
        &= \max\left\{{E_l(\psi), E_l(\phi)}\right\},
    \end{align}
    hence completing the proof, which holds $\forall l \in [1,d]$.
\end{proof}

We can now move to our first main result, namely the optimality of the thrifty protocol. Since the second phase of the thrifty protocol, converting $\ket{\psi \land \phi}$ into $\ket\phi$ (\textit{i.e.}, the third green arrow in Fig. \ref{fig:lattice_two_protocols}), is deterministic, all we need to consider is the first phase, namely the probabilistic LOCC transformation from $\ket\psi$ to $\ket{\psi \land \phi}$ (\textit{i.e.}, the two first green arrows in Fig. \ref{fig:lattice_two_protocols}). 
 In the following theorem, we prove that the maximum probability of achieving such a transformation from $\ket{\psi}$ to $\ket{\psi \land \phi}$ is equivalent to that of realizing the transformation from $\ket{\psi}$ to $\ket{\phi}$. In some sense, with this specific probability, one may as well stop the thrifty protocol halfway in order to produce the more entangled state $\ket{\psi \land \phi}$ rather than $\ket{\phi}$.

\begin{theorem}
\label{thm:optimal proba}
Let $\ket\psi$ and $\ket\phi$ be two bipartite pure states. The optimal probability of conversion from $\ket\psi$ to $\ket{\psi \land \phi}$, which we name $\mathfrak{r}_1$, is equal to that of the optimal conversion from $\ket\psi$ to $\ket\phi$, denoted as $r_1$ in Eq. (\ref{eq:r1}).
\end{theorem}

\begin{proof}
We need to show that $\mathfrak{r}_1 = r_1$, where $r_1$ is given by Eq. (\ref{eq:r1}) and $\mathfrak{r}_1$ is the optimal probability of conversion from $\ket\psi$ to $\ket{\psi \land \phi}$.
We know that this probabilistic transformation can be performed optimally using Vidal's protocol, so $\mathfrak{r}_1$ can be calculated with Eq. (\ref{eq:r1}), namely,
\begin{equation}
\label{eq:mathfrak r1}
    \mathfrak{r}_1 = \underset{l \in [1, d]}{\min}\frac{E_l(\psi)}{E_l(\psi \land \phi)} \equiv  \frac{E_{\mathfrak{l}_1}(\psi)}{E_{\mathfrak{l}_1}(\psi \land \phi)} ,
\end{equation}
where $\mathfrak{l}_1$ is defined as the smallest $l \in [1,d]$ minimizing $E_l(\psi)/E_l(\psi \land \phi)$.
With the use of Lemma \ref{lemma:entanglement monotones}, we can write
\begin{equation}
\label{eq:frac_r1}
    \mathfrak{r}_1 = \underset{l \in [1, d]}{\min} ~ \frac{E_l(\psi)}{\max\left\{{E_l(\psi), E_l(\phi)}\right\}} ,
\end{equation}
which is to be compared with $r_1$. We obviously have
\begin{equation}
    \frac{E_l(\psi)}{\max\left\{{E_l(\psi), E_l(\phi)}\right\}} \leq \frac{E_l(\psi)}{E_l(\phi)}, \quad \forall l \in [1,d].
\end{equation}
For any value of $l$ such that $ E_l(\phi) < E_l(\psi)$, we have
\begin{equation}
    1 = \frac{E_l(\psi)}{\max\left\{{E_l(\psi), E_l(\phi)}\right\}} < \frac{E_l(\psi)}{E_l(\phi)}, 
\end{equation}
implying that such a value of $l$ can be disregarded in the minimization over $l$ that yields $r_1$ in Eq. \eqref{eq:r1} or $\mathfrak{r}_1$ in Eq. \eqref{eq:frac_r1} since, when $l=1$, $E_1(\psi)= E_1(\phi)= 1$ so the minimum is guaranteed not to exceed 1. Thus, we may restrict the minimization over the values of $l$ such that $ E_l(\phi) \geq E_l(\psi)$, for which
\begin{equation}
\frac{E_l(\psi)}{\max\left\{{E_l(\psi), E_l(\phi)}\right\}} = \frac{E_l(\psi)}{E_l(\phi)}.
\end{equation}
This implies that we have equal conversion probabilities, 
\begin{equation}
    \mathfrak{r}_1 = r_1,
    \label{eq:equal-probabilities}
\end{equation}
hence we have proven that the thrifty protocol is also an optimal probabilistic protocol.



\end{proof}

We now turn to our second main result, namely that entanglement is better preserved, on average, by the thrifty protocol rather than the greedy protocol. Let us first prove the following lemma.

\begin{lemma}
\label{lemma:hadamard}
    For all $\textbf{x}$ and $\textbf{y} \in \mathcal{P}_d$ such that $\textbf{x} \prec \textbf{y}$ and for all decreasingly ordered vectors $\textbf{a} \in \mathbb{R}^d$ such that $\sum_{i=1}^d a_ix_i = \sum_{i=1}^d a_iy_i = 1$, we have
    \begin{equation}
        \textbf{a} \odot \textbf{x} \prec \textbf{a} \odot \textbf{y},
    \label{eq:lemma2}    
    \end{equation}
     where $\odot$ denotes the Hadamard (element-wise) product.
\end{lemma}

\begin{proof} 
The last equality [see Eq. \eqref{equality majorization}] in the majorization relation, Eq. \eqref{eq:lemma2}, is trivially fulfilled by hypothesis, so we need to prove the $d-1$ inequalities [see Eq. \eqref{eq:majo}] 
\begin{equation}
\label{eq:lemma hadamard majo relation}
       \sum_{i=1}^{k} a_i(y_i-x_i) \geq 0, \quad \forall k \in [1, d-1],
\end{equation}
given the majorization hypothesis $\textbf{\textit{x}} \prec \textbf{\textit{y}}$, which reads as
    \begin{equation}
        \sum_{i=1}^{k} (y_i-x_i) \geq 0 , \quad \forall k \in [1, d-1].
        \label{eq:majorization-hypothesis}
    \end{equation}
The decreasing rearrangement of vector \textbf{\textit{a}} allows us to write its components in the form
\begin{equation}
\label{eq:lemma2 a_i}
    a_i = \sum_{j=i}^d \mu_j, \quad \forall i \in [1, d],
\end{equation}
where $\mu_j \geq 0$, for all $j \in [1,d]$.
By plugging in Eq. (\ref{eq:lemma2 a_i}) into Eq. \eqref{eq:lemma hadamard majo relation}, we have to prove
\begin{equation}
\label{lemma1:eq1}
    \sum_{i=1}^{k} (y_i-x_i) \sum_{j=i}^d \mu_j  \geq 0, \quad \forall k \in [1, d-1],
\end{equation}
which can be rewritten by interchanging the two summation signs as
\begin{equation}
    \sum_{j=1}^d \mu_j \sum_{i=1}^{\min(k,j)} (y_i-x_i) \geq 0, \quad \forall k \in [1, d-1].
\end{equation}
These inequalities are clearly satisfied given $\mu_j \geq 0$ and the majorization hypothesis, Eq. \eqref{eq:majorization-hypothesis}, thereby completing the proof.

\end{proof}

With this in mind, we can prove our second main result, namely that the residual state $\ket\nu$ of the thrifty protocol is majorized by the residual state $\ket\xi$ of the greedy protocol. This is the content of the following theorem.

\begin{theorem}
\label{thm:better entanglement resource}
    Let $\ket\psi$ and $\ket\phi$ be two bipartite pure states and let $\ket\xi$ be the residual state if the two-outcome probabilistic protocol for converting $\ket\psi$ into $\ket\phi$ fails and $\ket\nu$ be the residual state if the two-outcome probabilistic protocol for converting $\ket\psi$ into $\ket{\psi \land \phi}$ fails. Both residual states are related via the majorization relation
    \begin{equation}
        \lambda_\nu \prec \lambda_\xi.
    \end{equation}
\end{theorem}

\begin{proof}

To show this, we proceed in three steps. First, using the definitions from Eq. (\ref{eq:ri}),
we characterize the sequence of ratios $\mathfrak{r}_j$'s (and corresponding values $\mathfrak{l}_j$'s) associated with the probabilistic conversion from $\ket\psi$ to $\ket{\psi\land\phi}$ and prove that $\mathfrak{r}_j=r_j$ and $\mathfrak{l}_j = l_j$, $\forall j$. Then, we express the intermediate state $\ket{\zeta}$ of this conversion 
and show that it satisfies $\lambda_\zeta \prec \lambda_\chi$, where $\ket{\chi}$ corresponds to the intermediate state of the probabilistic conversion from $\ket\psi$ (or from $\ket{\psi \lor \phi}$) to $\ket\phi$, as shown in Fig. \ref{fig:lattice_two_protocols}. Finally, we use $\lambda_\zeta \prec \lambda_\chi$ in order to prove the desired majorization relation $\lambda_\nu \prec \lambda_\xi$.

From Theorem \ref{thm:optimal proba}, we already know that for $j=1$, $\mathfrak{r}_1 = r_1$ and $\mathfrak{l}_1=l_1$. In order to prove the same relations for higher values of $j$, we proceed by iteration. We can write $\mathfrak{r}_2$, by means of Eq. (\ref{eq:ri}), as
\begin{align}
    \mathfrak{r}_2 &= \underset{l \in [1, \mathfrak{l}_1 - 1]}{\min}\frac{E_l(\psi) - E_{\mathfrak{l}_1}(\psi)}{E_l(\psi \land \phi)-E_{\mathfrak{l}_1}(\psi \land \phi)},\nonumber\\
    &\equiv \frac{E_{\mathfrak{l}_2}(\psi) - E_{\mathfrak{l}_1}(\psi)}{E_{\mathfrak{l}_2}(\psi \land \phi)-E_{\mathfrak{l}_1}(\psi \land \phi)},   
\end{align}
which can be reexpressed, using $l_1 = \mathfrak{l}_1$ and $E_{\mathfrak{l}_1}(\psi \land \phi) = \max\left\{E_{\mathfrak{l}_1}(\psi), E_{\mathfrak{l}_1}(\phi)\right\} = E_{l_1}(\phi)$, as
\begin{equation}
    \mathfrak{r}_2 = \underset{l \in [1, l_1 - 1]}{\min}\frac{E_l(\psi) - E_{l_1}(\psi)}{\max\left\{{E_l(\psi), E_l(\phi)}\right\}-E_{l_1}(\phi)}.
\end{equation}
We choose to rewrite this last expression as
\begin{equation}
    \mathfrak{r}_2 = \underset{l \in [1, l_1 - 1]}{\min} \mathfrak{r}_{2,l}.
\end{equation}
We now divide the values of $l$ in two categories. Either $l \in \mathcal{L}^+$ when $\max\left\{{E_l(\psi), E_l(\phi)}\right\} = E_{l}(\psi)$, or $l \in \mathcal{L}^-$ when $\max\left\{{E_l(\psi), E_l(\phi)}\right\} = E_{l}(\phi)$.
It is possible to prove that there always exists one value $l_{\min} \in \mathcal{L}^-$ such that $\mathfrak{r}_{2, l_{\min}} \leq \mathfrak{r}_{2, l}, \forall l \in \mathcal{L}^+ \cup \mathcal{L}^-$. Hence,
\begin{equation}
    \mathfrak{r}_2 = r_2.
\end{equation}
Proceeding equivalently for each $j$, we show that 
\begin{equation}
    \mathfrak{r}_j = r_j \text{ and }  \mathfrak{l}_j = l_j , \quad \forall j\in [1,k].
\end{equation}
This allows us to write the vector of Schmidt coefficients of the intermediate state $\ket{\zeta}$ between $\ket{\psi}$ and $\ket{\psi \land \phi}$ as
\begin{equation}
    \lambda_\zeta^{(i)} = r_j \, \lambda_{\psi \land \phi}^{(i)}, \quad \text{if} \ i \in [l_j, l_{j-1}-1],\quad j \in [1,k].
\end{equation}
Thus, this vector can be expressed as a Hadamard product of two vectors, namely,
\begin{equation}
    \lambda_\zeta = \textbf{\textit{r}} \odot \lambda_{\psi \land \phi},
\end{equation}
where $(\textbf{\textit{r}})_i = r_j$, if $i \in [l_j, l_{j-1}-1], \forall i \in [1,d]$. 
Similarly, in view of Eq. (\ref{eq:approximate state Vidal}), the vector of Schmidt coefficients of the intermediate state $\ket{\chi}$ between $\ket{\psi}$ (or $\ket{\psi \lor \phi}$) and $\ket{\phi}$ can be rewritten   using this notation as
\begin{equation}
    \lambda_\chi = \textbf{\textit{r}} \odot \lambda_{\phi}.
\end{equation}
Now, we simply make use of Lemma \ref{lemma:hadamard}.
Since $\lambda_{\psi \land \phi} \prec \lambda_\phi$, the majorization relation
\begin{equation}
\lambda_\zeta = \textbf{\textit{r}} \odot \lambda_{\psi \land \phi} \prec \textbf{\textit{r}} \odot \lambda_{\phi} = \lambda_\chi,
\label{eq-majorization-zeta-chi}
\end{equation}
holds true provided the vector $\textbf{\textit{r}}$ is decreasingly ordered (by increasing index $i$). This can be easily checked from the structure of this vector, namely,
\begin{equation}
\textbf{\textit{r}} = (
\underset{\overset{\uparrow}{ l_k}}{r_k} \hdots  r_k | \hdots\hdots | \underset{\overset{\uparrow}{ l_2}}{r_2} \hdots r_2 | \underset{\overset{\uparrow}{ l_1}}{r_1} \hdots r_1 ),
\end{equation}
where each $l_j$ variable points to the index of the corresponding element in $\textbf{\textit{r}}$. 
Hence, $\lambda_\zeta \prec \lambda_\chi$, which means that the intermediate state $\ket{\zeta}$ of the thrifty protocol is more entangled than the intermediate state $\ket{\chi}$ of the greedy protocol, as can also be seen in Fig. \ref{fig:lattice_two_protocols}.

Finally, we deduce from Eq. \eqref {eq-majorization-zeta-chi} the fact that $\lambda_\nu \prec \lambda_\xi$, \textit{i.e.}, that if the protocol fails, the residual state of the thrifty protocol is more entangled than the residual state of the greedy protocol (see inset of Fig. \ref{fig:lattice_two_protocols}). To do this, first notice that 
\begin{equation}
\label{xi from measurement}
    \hat N \otimes \hat I \ket{\chi} = \sqrt{1-r_1} \ket{\xi},
\end{equation}
and 
\begin{equation}
\label{nu from measurement}
    \hat N \otimes \hat I \ket{\zeta} = \sqrt{1-r_1} \ket{\nu}.
\end{equation}
Given that $\hat N$ is diagonal [see Eq. (\ref{N})], we can rewrite Eqs. (\ref{xi from measurement}) and (\ref{nu from measurement}) using vectors of Schmidt coefficients and the Hadamard product as
\begin{align}
    \lambda_\xi = \textbf{\textit{n}} \odot \lambda_\chi,\\
    \lambda_\nu = \textbf{\textit{n}} \odot \lambda_\zeta,
\end{align}
where $(\textbf{\textit{n}})_i = \langle i | \hat N | i \rangle / \sqrt{1-r_1}, \forall i \in [1,d]$.
Finally, using again Lemma \ref{lemma:hadamard}, valid because of the fact that $\textbf{\textit{n}}$ is a decreasingly ordered vector and the fact that $\lambda_\zeta \prec \lambda_\chi$, we complete the proof that $\lambda_\nu \prec \lambda_\xi$.

\end{proof}

\section{Generalization to a collection of initial or final states}
\label{sec:generalization to multiple final states}

Let us first generalize Theorem \ref{thm:optimal proba} to a scenario involving multiple final states. Imagine that Alice and Bob possess a state $\ket\psi$ and that, instead of a single target state, there is a collection of $m$ possible target states $\{\ket{\phi_j}\}_{j=1}^m$ from which they will have to produce a single one (without knowing which one beforehand). The idea is that they can first perform a probabilistic LOCC transformation to the OCR state of the $m+1$ states, denoted as
\begin{equation}
    \ket{\psi \overset{m}{\underset{j=1}{\land}} \phi_j} \equiv \ket{\psi \land \phi_1\land \phi_2\land \cdots \land \phi_m},
\end{equation}
extending the essence of the thrifty protocol. Once they possess the OCR state, they can wait until they learn which state $\ket{\phi_j}$ they have to produce before performing the corresponding deterministic LOCC transformation from the OCR state to $\ket{\phi_j}$. Interestingly, the probabilistic LOCC from $\ket\psi$ to the OCR state has an optimal probability that is equal to the minimum optimal probability of all individual conversions from $\ket\psi$ to $\ket{\phi_j}$. In some sense, with this minimum probability corresponding to the hardest-to-reach target state $\ket{\phi_j}$, we may as well stop the protocol halfway to produce the more entangled OCR state. This is the content of the following theorem.


\begin{theorem}
\label{thm:last theorem}
Let $\ket\psi$ be a bipartite pure state 
and $\{\ket{\phi_j}\}_{j=1}^m$ be a collection of $m$ bipartite pure states. The optimal probability $p_{\max}^{[1]}$ of conversion from $\ket\psi$ to the OCR state of the $m+1$ states $\ket{\psi \overset{m}{\underset{j=1}{\land}} \phi_j}$ is equal to the minimum optimal probability among the $m$ possible individual conversions from $\ket\psi$ to $\ket{\phi_j}$, \textit{i.e.}, 
\begin{equation}
    p_{\max}^{[1]} = \underset{j\in [1,m]}{\min}\ p_{\max,j}^{[2]},
\end{equation}
where $p_{\max,j}^{[2]}$ denotes the optimal probability of conversion from $\ket\psi$ to $\ket{\phi_j}$.
\end{theorem}

\begin{proof}
From Eq. (\ref{eq:optimal proba conversion}), we have
\begin{equation}
    p_{\max}^{[1]} = \underset{l \in [1,d]}{\min} \frac{E_l(\psi)}{E_l(\psi \overset{m}{\underset{j=1}{\land}} \phi_j)},
\end{equation}
as well as
\begin{equation}
     p_{\max,j}^{[2]} = \underset{l \in [1,d]}{\min} \frac{E_l(\psi)}{E_l(\phi_j)},
\end{equation}
and we want to prove that
\begin{equation}
    p_{\max}^{[1]} = \underset{j \in [1,m]}{\min} \, p_{\max,j}^{[2]} \equiv p_{\max}^{[2]}.
\end{equation}

In order to do so, we can use the straightforward generalization of Lemma \ref{lemma:entanglement monotones}, namely,
\begin{equation}
    E_l(\psi \overset{m}{\underset{j=1}{\land}} \phi_j) = \max \left\{ E_l(\psi), E_l(\phi_1), \cdots, E_l(\phi_m)\right\}.
\end{equation}
Thus we have $p_{\max}^{[1]}=\underset{l\in[1,d]}{\min} \, p_{\max,l}^{[1]}$, where 
\begin{equation}
\label{eq:appendix pmax,l 1}
p_{\max,l}^{[1]} = \frac{E_l(\psi)}{\max \left\{ E_l(\psi), E_l(\phi_1), \cdots, E_l(\phi_m)\right\}},
\end{equation} 
and, by interchanging the minima, we can also write 
\begin{align}
     p_{\max}^{[2]} 
     &= \underset{l \in [1,d]}{\min} p_{\max,l}^{[2]},
\end{align}
where 
\begin{equation}
\label{eq:appendix pmax,l 2}
p_{\max,l}^{[2]} = \underset{j \in [1,m]}{\min} \frac{E_l(\psi)}{E_l(\phi_j)}.
\end{equation} 
Therefore, using Eqs. (\ref{eq:appendix pmax,l 1}) and (\ref{eq:appendix pmax,l 2}), it is easy to see that $p_{\max,l}^{[2]} \geq p_{\max,l}^{[1]} = 1$ for any value of $l$ such that  $\max \left\{ E_l(\psi), E_l(\phi_1), \cdots, E_l(\phi_m)\right\} = E_l(\psi)$. These values of $l$ can thus be disregarded in the minimization over $l$ since the minimum cannot exceed 1 (indeed,  $p_{\max,1}^{[2]} = p_{\max,1}^{[1]} = 1$ when $l=1$). Thus, we only have to consider the values of $l$ such that $\max \left\{ E_l(\psi), E_l(\phi_1), \cdots, E_l(\phi_m)\right\} = E_l(\phi_k)$, $k\in[1,m]$, implying  $p_{\max,l}^{[2]} = p_{\max,l}^{[1]} = E_l(\psi)/E_l(\phi_k)$.
Hence, the minimum over $l$ yields $p_{\max}^{[1]} = p_{\max}^{[2]}$,
which completes the proof.
\end{proof}

In other words, as long as Alice and Bob do not know which target state $\ket{\phi_j}$ they will have to reach, they may anticipate the probabilistic step and move to the OCR, but this can only be done with the minimum probability (corresponding to the worst-case target state). The benefit of this procedure is that once the target $\ket{\phi_j}$ is disclosed, the final step is deterministic.
Otherwise, they may wait until the target $\ket{\phi_j}$ is disclosed and only then perform a probabilistic transformation from $\ket{\psi}$ to $\ket{\phi_j}$, whose probability may be higher (depending on which is the target state). In a nutshell, a compromise can be made between doing the probabilistic transformation at first (making the rest fully deterministic regardless of the target state) or waiting to know which target state is wanted before doing the actual probabilistic transformation (possibly with a higher success probability). 

It is worth noting that the Procrustean method introduced in Ref.  \cite{BennettBernstein_Concentrating_1996} can be viewed as a special case of our thrifty protocol in a scenario where the target state is \textit{a priori} totally arbitrary. The Procrustean method is a local filtering which converts a single copy of a given partially entangled pure state $\ket{\psi}$ into the maximally entangled state $\ket{\Phi}$ of the same dimension. Indeed, when considering the thrifty protocol with the set of target states $\{\ket{\phi_j}\}$ being the whole set of states of the same dimension as the initial state, the OCR then becomes the maximally entangled state $\ket{\Phi}$. Hence, the thrifty protocol starts with the probabilistic conversion of $\ket{\psi}$ into $\ket{\Phi}$, which is thus equivalent to the Procrustean method\footnote{This is a consequence of the fact that Vidal's protocol itself can be seen as a generalization of the Procrustean method when the target state differs from the maximally entangled state of the same dimension as the initial state.}. Once the actual target state $\ket{\phi_j}$ is revealed, one may then deterministically convert $\ket{\Phi}$ into $\ket{\phi_j}$ (which is possible since $\ket{\Phi} \prec \ket{\phi_j}$, $\forall j$). Of course, this Procrustean procedure becomes suboptimal in situations where the set of target states is limited so the OCR is not maximally entangled, in which case the thrifty protocol comes with a higher success probability. 

It is also tempting to compare our multitarget-state thrifty protocol to the protocol introduced in Ref. \cite{JonathanPlenio_MinimalConditions_1999}. There, the conversion into an ensemble of pure states with associated probabilities is considered. Specifically, the conversion of $\ket{\psi}$ into the ensemble  $\{q_j,\ket{\phi_j} \}_{j=1}^m$ is shown to be possible via a LOCC transformation if and only if
\begin{equation}
\label{eq:jonathanplenio}
    \lambda_\psi \prec \sum_{j=1}^m q_j \lambda_{\phi_j},
\end{equation}
which extends Eq. \eqref{eq:nielsen majorization}. The conversion starts with a deterministic step towards an ``average" pure state whose Schmidt vector is equal to $\sum_{j=1}^m q_j \lambda_{\phi_j}$, followed by a probabilistic step yielding each state $\ket{\phi_j}$ with probability $q_j$.
Despite the apparent similarity between such a protocol and our multi-target thrifty protocol, the two have quite different purposes. Indeed, our protocol has a specific probability of succeeding, \textit{i.e.}, outputting the OCR, which is then deterministically converted into one of the desired target states $\ket{\phi_j}$ chosen at will. According to Theorem \ref{thm:last theorem}, the success probability is the minimum of the probabilities associated with each individual conversion (Vidal’s bound for each pair of initial and target states). In contrast, the protocol of Ref. \cite{JonathanPlenio_MinimalConditions_1999} is designed as a source that randomly outputs a state $\ket{\phi_j}$ drawn from a target ensemble. The weights $q_j$ in this ensemble can be chosen freely but must satisfy Eq. \eqref{eq:jonathanplenio} for any physically realizable ensemble. In particular, each $q_j$ must still be upper bounded by Vidal’s corresponding optimal conversion probability from $\ket\psi$ to $\ket{\phi_j}$.

To complete the picture, we may also consider a scenario involving multiple initial states $\{\ket{\psi_j}\}_{j=1}^m$ and a single target state $\ket\phi$, generalizing the greedy protocol. This time, the optimal probability of conversion $p_{\max}^{[1]}$ from the OCP state of the $m+1$ states $\ket{\phi \overset{m}{\underset{j=1}{\lor}}\psi_j}$\footnote{In analogy with the OCR, the OCP of multiple states can be understood as $\ket{\phi \overset{m}{\underset{j=1}{\lor}} \psi_j} \equiv \ket{\phi \lor \psi_1\lor \psi_2\lor \cdots \lor \psi_m}$.} to the target state $\ket\phi$ is equivalent to the minimum optimal probability between any of the $m$ possible individual conversions from $\ket{\psi_j}$ to $\ket{\phi}$. This is immediately understood because, as proven in Ref. \cite{BosykEtAl_ApproximateTransformationsOfBipartite_2017}, the intermediate state for the least probable of the $m$ possible conversions is located in the lower cone of $\ket{\phi \overset{m}{\underset{j=1}{\lor}}\psi_j}$. Therefore, the transformation from the OCP state $\ket{\phi \overset{m}{\underset{j=1}{\lor}}\psi_j}$ to $\ket{\phi}$ cannot be more probable than the least probable of the $m$ possible conversions from $\ket{\psi_j}$ to $\ket{\phi}$. Unlike the multi-target thrifty protocol, this multi-initial-state generalization of the greedy protocol does not  have a straightforward practical meaning since the deterministic transformation from some of the initial states $\ket{\psi_j}$ to the OCP state cannot be performed until the actual $\ket{\psi_j}$ is disclosed. However, it is relevant in a scenario involving a catalyst state \cite{JonathanPlenio_Catalyst_1999}. If one postpones the deterministic step and, instead, borrows the OCP state, then it is possible to implement the probabilistic step yielding $\ket{\phi}$ from the OCP state. Later on, once the actual $\ket{\psi_j}$ is disclosed, it can be deterministically converted into the OCP state, which is returned and plays therefore the role of a catalyst. Here again, there is a compromise between doing the probabilistic step at first (postponing the subsequent deterministic step) or waiting until the identity of the initial state is known before doing the probabilistic transformation (possibly with a higher success probability). 



\section{Conclusion}
\label{sec:conclusion}


The lattice structure of majorization uncovers two essential states in the context of entanglement transformations between incomparable bipartite pure states $\ket{\psi}$ and $\ket{\phi}$, namely the optimal common resource state (\textit{i.e.}, the \textit{meet} $\ket{\psi\land \phi}$ of the two states) and the optimal common product state (\textit{i.e.}, the \textit{join} $\ket{\psi\lor \phi}$ of the two states). We have shown that both states naturally appear when considering (single-copy) probabilistic LOCC transformations from $\ket{\psi}$ to $\ket{\phi}$. We have indeed defined two antipodal protocols, namely the greedy protocol, passing through $\ket{\psi\lor \phi}$, and the thrifty protocol, passing through $\ket{\psi\land \phi}$. Both protocols can be proven to be optimal (their success probability is maximum). However, while the greedy protocol is very similar to Vidal's protocol \cite{Vidal_EntanglementOfPureStatesForASingleCopy_1999}, the thrifty protocol is superior in that the entanglement resource is better preserved on average (in case of failure, its residual state is majorized by -- hence, is more entangled than -- the residual state of the greedy protocol).  Note that in case $\ket{\psi}$ and $\ket{\phi}$ are comparable, 
both the greedy and thrifty protocols reduce to  Vidal's protocol, which reflects that the incomparability between the states is an essential ingredient here.
Finally, we have shown that the greedy and thrifty protocols can be generalized to scenarios involving an arbitrary number of initial or final states. This underlines the operational relevance of the majorization lattice in the scope of quantum entanglement theory.

Overall, the current work sheds a new light on the resource theory of entanglement. Such a theory has proven to be essential in numerous areas within quantum information sciences. For example, the rapidly developing field of quantum networks \cite{Perseguers_QuantumRandomNetworks_2010, Malik_ConcurrencePercolation_2022} relies on the ability to convert and manipulate entanglement at the local scale so to establish useful entanglement between distant nodes. Another notable example is quantum thermodynamics \cite{HorodeckiOppenheim_FundamentalLimitations_2013, Brandao_TheSecondLaws_2015, Singh_PartialOrder_2021}, where the resource-theoretical approach is well suited to describe the allowed state transformations (here, thermo-majorization is the condition that reflects the existence of thermal operations) and where the notion of majorization lattice has already proven its relevance \cite{Korzekwa_StructureOfTheThermodynamicArrow_2017, deOliveiraJuniorEtAl_GeometricStructure_2022}. Thus, plugging the resource theory of entanglement into a majorization lattice, as we have sketched here, can be expected to be a very promising research avenue. We may conceive cryptographic or thermodynamic scenarios where the convertibility (or non-convertibility) from or towards the meet or join state is crucial.

Finally, it would be interesting to extend the current work to multipartite settings \cite{Neven_LocalTransformations_2021, ChitambarDuanShi_Tripartite_2008}. To our knowledge, the notions of optimal common resource or optimal common product have never been explored for more than two parties. For example, we know that true tripartite entanglement can be split into two categories, namely GHZ-type and W-type states \cite{DurVidalCirac_ThreeQubitsCanBeEntangled_2000}, but the determination of an optimal common resource state or optimal common product state within each class constitutes an open problem. In another direction, extending the present analysis to the conversion of mixed states and exploring the role of incomparability in this context would, if possible, even further broaden the applicability of this framework.




\acknowledgments
We warmly thank Ali Asadian for pointing out to us the relevance of the majorization lattice in the context of entanglement transformations as well as Michael Jabbour for useful discussions on majorization theory. We also acknowledge useful comments from an anonymous referee. S.D. is a FRIA grantee of the Fonds de la Recherche Scientifique – FNRS. M.A. acknowledges support from the European Union’s HORIZON Research and Innovation Actions under Grant Agreement No. 101080173 (CLUSTEC) and the European Union’s 2020 research and innovation programme (CSA - Coordination and support action, 951737H2020-WIDESPREAD-2020-5) under Grant Agreement No. 951737 (NONGAUSS) as well as support from MEYS Czech Republic and the European Union’s Horizon 2020 (2014–2020) research and innovation framework programme under Grant No. 731473 (Project No. 8C20002 ShoQC). C.G. acknowledges support from the Fonds de la Recherche Scientifique – FNRS. N.J.C. acknowledges support from the Fonds de la Recherche Scientifique – FNRS and the European Union under project ShoQC within the Horizon 2020 ERA-NET Cofund in Quantum Technologies (QuantERA) program.

\bibliography{main}

\begin{thebibliography}{32}%
\makeatletter
\providecommand \@ifxundefined [1]{%
 \@ifx{#1\undefined}
}%
\providecommand \@ifnum [1]{%
 \ifnum #1\expandafter \@firstoftwo
 \else \expandafter \@secondoftwo
 \fi
}%
\providecommand \@ifx [1]{%
 \ifx #1\expandafter \@firstoftwo
 \else \expandafter \@secondoftwo
 \fi
}%
\providecommand \natexlab [1]{#1}%
\providecommand \enquote  [1]{``#1''}%
\providecommand \bibnamefont  [1]{#1}%
\providecommand \bibfnamefont [1]{#1}%
\providecommand \citenamefont [1]{#1}%
\providecommand \href@noop [0]{\@secondoftwo}%
\providecommand \href [0]{\begingroup \@sanitize@url \@href}%
\providecommand \@href[1]{\@@startlink{#1}\@@href}%
\providecommand \@@href[1]{\endgroup#1\@@endlink}%
\providecommand \@sanitize@url [0]{\catcode `\\12\catcode `\$12\catcode
  `\&12\catcode `\#12\catcode `\^12\catcode `\_12\catcode `\%12\relax}%
\providecommand \@@startlink[1]{}%
\providecommand \@@endlink[0]{}%
\providecommand \url  [0]{\begingroup\@sanitize@url \@url }%
\providecommand \@url [1]{\endgroup\@href {#1}{\urlprefix }}%
\providecommand \urlprefix  [0]{URL }%
\providecommand \Eprint [0]{\href }%
\providecommand \doibase [0]{https://doi.org/}%
\providecommand \selectlanguage [0]{\@gobble}%
\providecommand \bibinfo  [0]{\@secondoftwo}%
\providecommand \bibfield  [0]{\@secondoftwo}%
\providecommand \translation [1]{[#1]}%
\providecommand \BibitemOpen [0]{}%
\providecommand \bibitemStop [0]{}%
\providecommand \bibitemNoStop [0]{.\EOS\space}%
\providecommand \EOS [0]{\spacefactor3000\relax}%
\providecommand \BibitemShut  [1]{\csname bibitem#1\endcsname}%
\let\auto@bib@innerbib\@empty
\bibitem [{\citenamefont {Horodecki}\ \emph {et~al.}(2009)\citenamefont
  {Horodecki}, \citenamefont {Horodecki}, \citenamefont {Horodecki},\ and\
  \citenamefont {Horodecki}}]{Horodeckis_QuantumEntanglement_2009}%
  \BibitemOpen
  \bibfield  {author} {\bibinfo {author} {\bibfnamefont {R.}~\bibnamefont
  {Horodecki}}, \bibinfo {author} {\bibfnamefont {P.}~\bibnamefont
  {Horodecki}}, \bibinfo {author} {\bibfnamefont {M.}~\bibnamefont
  {Horodecki}},\ and\ \bibinfo {author} {\bibfnamefont {K.}~\bibnamefont
  {Horodecki}},\ }\bibfield  {title} {\bibinfo {title} {{Quantum
  entanglement}},\ }\href {https://doi.org/10.1103/RevModPhys.81.865}
  {\bibfield  {journal} {\bibinfo  {journal} {Rev. Mod. Phys.}\ }\textbf
  {\bibinfo {volume} {81}},\ \bibinfo {pages} {865} (\bibinfo {year}
  {2009})}\BibitemShut {NoStop}%
\bibitem [{\citenamefont {Bennett}\ and\ \citenamefont
  {Wiesner}(1992)}]{BennettWiesner_CommunicationViaOneAndTwoParticleOperators_1992}%
  \BibitemOpen
  \bibfield  {author} {\bibinfo {author} {\bibfnamefont {C.~H.}\ \bibnamefont
  {Bennett}}\ and\ \bibinfo {author} {\bibfnamefont {S.~J.}\ \bibnamefont
  {Wiesner}},\ }\bibfield  {title} {\bibinfo {title} {Communication via one-
  and two-particle operators on {Einstein}-{Podolsky}-{Rosen} states},\ }\href
  {https://doi.org/10.1103/PhysRevLett.69.2881} {\bibfield  {journal} {\bibinfo
   {journal} {Phys. Rev. Lett.}\ }\textbf {\bibinfo {volume} {69}},\ \bibinfo
  {pages} {2881} (\bibinfo {year} {1992})}\BibitemShut {NoStop}%
\bibitem [{\citenamefont {Bennett}\ \emph {et~al.}(1993)\citenamefont
  {Bennett}, \citenamefont {Brassard}, \citenamefont {Cr\'epeau}, \citenamefont
  {Jozsa}, \citenamefont {Peres},\ and\ \citenamefont
  {Wootters}}]{BennettEtAl_TeleportingAnUnknownQuantumState_1993}%
  \BibitemOpen
  \bibfield  {author} {\bibinfo {author} {\bibfnamefont {C.~H.}\ \bibnamefont
  {Bennett}}, \bibinfo {author} {\bibfnamefont {G.}~\bibnamefont {Brassard}},
  \bibinfo {author} {\bibfnamefont {C.}~\bibnamefont {Cr\'epeau}}, \bibinfo
  {author} {\bibfnamefont {R.}~\bibnamefont {Jozsa}}, \bibinfo {author}
  {\bibfnamefont {A.}~\bibnamefont {Peres}},\ and\ \bibinfo {author}
  {\bibfnamefont {W.~K.}\ \bibnamefont {Wootters}},\ }\bibfield  {title}
  {\bibinfo {title} {Teleporting an unknown quantum state via dual classical
  and {Einstein}-{Podolsky}-{Rosen} channels},\ }\href
  {https://doi.org/10.1103/PhysRevLett.70.1895} {\bibfield  {journal} {\bibinfo
   {journal} {Phys. Rev. Lett.}\ }\textbf {\bibinfo {volume} {70}},\ \bibinfo
  {pages} {1895} (\bibinfo {year} {1993})}\BibitemShut {NoStop}%
\bibitem [{\citenamefont {Chitambar}\ and\ \citenamefont
  {Gour}(2019)}]{ChitambarGour_QuantumResourceTheories_2019}%
  \BibitemOpen
  \bibfield  {author} {\bibinfo {author} {\bibfnamefont {E.}~\bibnamefont
  {Chitambar}}\ and\ \bibinfo {author} {\bibfnamefont {G.}~\bibnamefont
  {Gour}},\ }\bibfield  {title} {\bibinfo {title} {Quantum resource theories},\
  }\href {https://doi.org/10.1103/RevModPhys.91.025001} {\bibfield  {journal}
  {\bibinfo  {journal} {Rev. Mod. Phys.}\ }\textbf {\bibinfo {volume} {91}},\
  \bibinfo {pages} {025001} (\bibinfo {year} {2019})}\BibitemShut {NoStop}%
\bibitem [{\citenamefont {Bennett}\ \emph {et~al.}(1996)\citenamefont
  {Bennett}, \citenamefont {Bernstein}, \citenamefont {Popescu},\ and\
  \citenamefont {Schumacher}}]{BennettBernstein_Concentrating_1996}%
  \BibitemOpen
  \bibfield  {author} {\bibinfo {author} {\bibfnamefont {C.~H.}\ \bibnamefont
  {Bennett}}, \bibinfo {author} {\bibfnamefont {H.~J.}\ \bibnamefont
  {Bernstein}}, \bibinfo {author} {\bibfnamefont {S.}~\bibnamefont {Popescu}},\
  and\ \bibinfo {author} {\bibfnamefont {B.}~\bibnamefont {Schumacher}},\
  }\bibfield  {title} {\bibinfo {title} {Concentrating partial entanglement by
  local operations},\ }\href {https://doi.org/10.1103/PhysRevA.53.2046}
  {\bibfield  {journal} {\bibinfo  {journal} {Phys. Rev. A}\ }\textbf {\bibinfo
  {volume} {53}},\ \bibinfo {pages} {2046} (\bibinfo {year}
  {1996})}\BibitemShut {NoStop}%
\bibitem [{\citenamefont {Nielsen}(1999)}]{Nielsen_Conditions_1999}%
  \BibitemOpen
  \bibfield  {author} {\bibinfo {author} {\bibfnamefont {M.~A.}\ \bibnamefont
  {Nielsen}},\ }\bibfield  {title} {\bibinfo {title} {Conditions for a {Class}
  of {Entanglement} {Transformations}},\ }\href
  {https://doi.org/10.1103/PhysRevLett.83.436} {\bibfield  {journal} {\bibinfo
  {journal} {Phys. Rev. Lett.}\ }\textbf {\bibinfo {volume} {83}},\ \bibinfo
  {pages} {436} (\bibinfo {year} {1999})}\BibitemShut {NoStop}%
\bibitem [{\citenamefont
  {Vidal}(1999)}]{Vidal_EntanglementOfPureStatesForASingleCopy_1999}%
  \BibitemOpen
  \bibfield  {author} {\bibinfo {author} {\bibfnamefont {G.}~\bibnamefont
  {Vidal}},\ }\bibfield  {title} {\bibinfo {title} {Entanglement of {Pure}
  {States} for a {Single} {Copy}},\ }\href
  {https://doi.org/10.1103/PhysRevLett.83.1046} {\bibfield  {journal} {\bibinfo
   {journal} {Phys. Rev. Lett.}\ }\textbf {\bibinfo {volume} {83}},\ \bibinfo
  {pages} {1046} (\bibinfo {year} {1999})}\BibitemShut {NoStop}%
\bibitem [{\citenamefont {Marshall}\ \emph {et~al.}(2011)\citenamefont
  {Marshall}, \citenamefont {Olkin},\ and\ \citenamefont
  {Arnold}}]{MarshallOlkinArnold_Inequalities_2011}%
  \BibitemOpen
  \bibfield  {author} {\bibinfo {author} {\bibfnamefont {A.~W.}\ \bibnamefont
  {Marshall}}, \bibinfo {author} {\bibfnamefont {I.}~\bibnamefont {Olkin}},\
  and\ \bibinfo {author} {\bibfnamefont {B.~C.}\ \bibnamefont {Arnold}},\
  }\bibfield  {title} {\bibinfo {title} {Inequalities: {Theory} of
  {Majorization} and {Its} {Applications}},\ }\href
  {https://link.springer.com/book/10.1007/978-0-387-68276-1} {\bibfield
  {journal} {\bibinfo  {journal} {Springer New York}\ } (\bibinfo {year}
  {2011})}\BibitemShut {NoStop}%
\bibitem [{\citenamefont {Vidal}(2000)}]{Vidal_EntanglementMonotones_2000}%
  \BibitemOpen
  \bibfield  {author} {\bibinfo {author} {\bibfnamefont {G.}~\bibnamefont
  {Vidal}},\ }\bibfield  {title} {\bibinfo {title} {Entanglement monotones},\
  }\href {https://doi.org/10.1080/09500340008244048} {\bibfield  {journal}
  {\bibinfo  {journal} {Journal of Modern Optics}\ }\textbf {\bibinfo {volume}
  {47}},\ \bibinfo {pages} {355} (\bibinfo {year} {2000})}\BibitemShut
  {NoStop}%
\bibitem [{\citenamefont {Jensen}\ and\ \citenamefont
  {Schack}(2001)}]{Jensen_SimpleAlgorithmForLocalConversion_2001}%
  \BibitemOpen
  \bibfield  {author} {\bibinfo {author} {\bibfnamefont {J.~G.}\ \bibnamefont
  {Jensen}}\ and\ \bibinfo {author} {\bibfnamefont {R.}~\bibnamefont
  {Schack}},\ }\bibfield  {title} {\bibinfo {title} {Simple algorithm for local
  conversion of pure states},\ }\href
  {https://doi.org/10.1103/PhysRevA.63.062303} {\bibfield  {journal} {\bibinfo
  {journal} {Phys. Rev. A}\ }\textbf {\bibinfo {volume} {63}},\ \bibinfo
  {pages} {062303} (\bibinfo {year} {2001})}\BibitemShut {NoStop}%
\bibitem [{\citenamefont {Torun}\ and\ \citenamefont
  {Yildiz}(2015)}]{TorunYildiz_DeterministicTransformations_2015}%
  \BibitemOpen
  \bibfield  {author} {\bibinfo {author} {\bibfnamefont {G.}~\bibnamefont
  {Torun}}\ and\ \bibinfo {author} {\bibfnamefont {A.}~\bibnamefont {Yildiz}},\
  }\bibfield  {title} {\bibinfo {title} {Deterministic transformations of
  bipartite pure states},\ }\href
  {https://doi.org/10.1016/j.physleta.2014.11.013} {\bibfield  {journal}
  {\bibinfo  {journal} {Phys. Lett. A}\ }\textbf {\bibinfo {volume} {379}},\
  \bibinfo {pages} {113} (\bibinfo {year} {2015})}\BibitemShut {NoStop}%
\bibitem [{\citenamefont {Cicalese}\ and\ \citenamefont
  {Vaccaro}(2002)}]{CicaleseVaccaro_Supermodularity_2002}%
  \BibitemOpen
  \bibfield  {author} {\bibinfo {author} {\bibfnamefont {F.}~\bibnamefont
  {Cicalese}}\ and\ \bibinfo {author} {\bibfnamefont {U.}~\bibnamefont
  {Vaccaro}},\ }\bibfield  {title} {\bibinfo {title} {Supermodularity and
  subadditivity properties of the entropy on the majorization lattice},\ }\href
  {https://doi.org/10.1109/18.992785} {\bibfield  {journal} {\bibinfo
  {journal} {IEEE Transactions on Information Theory}\ }\textbf {\bibinfo
  {volume} {48}},\ \bibinfo {pages} {933} (\bibinfo {year} {2002})}\BibitemShut
  {NoStop}%
\bibitem [{\citenamefont
  {Korzekwa}(2017)}]{Korzekwa_StructureOfTheThermodynamicArrow_2017}%
  \BibitemOpen
  \bibfield  {author} {\bibinfo {author} {\bibfnamefont {K.}~\bibnamefont
  {Korzekwa}},\ }\bibfield  {title} {\bibinfo {title} {Structure of the
  thermodynamic arrow of time in classical and quantum theories},\ }\href
  {https://doi.org/10.1103/PhysRevA.95.052318} {\bibfield  {journal} {\bibinfo
  {journal} {Phys. Rev. A}\ }\textbf {\bibinfo {volume} {95}},\ \bibinfo
  {pages} {052318} (\bibinfo {year} {2017})}\BibitemShut {NoStop}%
\bibitem [{\citenamefont {Bosyk}\ \emph {et~al.}(2017)\citenamefont {Bosyk},
  \citenamefont {Sergioli}, \citenamefont {Freytes}, \citenamefont {Holik},\
  and\ \citenamefont
  {Bellomo}}]{BosykEtAl_ApproximateTransformationsOfBipartite_2017}%
  \BibitemOpen
  \bibfield  {author} {\bibinfo {author} {\bibfnamefont {G.~M.}\ \bibnamefont
  {Bosyk}}, \bibinfo {author} {\bibfnamefont {G.}~\bibnamefont {Sergioli}},
  \bibinfo {author} {\bibfnamefont {H.}~\bibnamefont {Freytes}}, \bibinfo
  {author} {\bibfnamefont {F.}~\bibnamefont {Holik}},\ and\ \bibinfo {author}
  {\bibfnamefont {G.}~\bibnamefont {Bellomo}},\ }\bibfield  {title} {\bibinfo
  {title} {Approximate transformations of bipartite pure-state entanglement
  from the majorization lattice},\ }\href
  {https://doi.org/10.1016/j.physa.2016.12.083} {\bibfield  {journal} {\bibinfo
   {journal} {Physica A: Statistical Mechanics and its Applications}\ }\textbf
  {\bibinfo {volume} {473}},\ \bibinfo {pages} {403} (\bibinfo {year}
  {2017})}\BibitemShut {NoStop}%
\bibitem [{\citenamefont {Bosyk}\ \emph {et~al.}(2018)\citenamefont {Bosyk},
  \citenamefont {Freytes}, \citenamefont {Bellomo},\ and\ \citenamefont
  {Sergioli}}]{BosykFreytesBellomo_LatticeOfTrumpingMajorization_2018}%
  \BibitemOpen
  \bibfield  {author} {\bibinfo {author} {\bibfnamefont {G.~M.}\ \bibnamefont
  {Bosyk}}, \bibinfo {author} {\bibfnamefont {H.}~\bibnamefont {Freytes}},
  \bibinfo {author} {\bibfnamefont {G.}~\bibnamefont {Bellomo}},\ and\ \bibinfo
  {author} {\bibfnamefont {G.}~\bibnamefont {Sergioli}},\ }\bibfield  {title}
  {\bibinfo {title} {The lattice of trumping majorization for {4D} probability
  vectors and {2D} catalysts},\ }\href
  {https://doi.org/10.1038/s41598-018-21947-0} {\bibfield  {journal} {\bibinfo
  {journal} {Sci. Rep.}\ }\textbf {\bibinfo {volume} {8}},\ \bibinfo {pages}
  {3671} (\bibinfo {year} {2018})}\BibitemShut {NoStop}%
\bibitem [{\citenamefont {Bosyk}\ \emph {et~al.}(2019)\citenamefont {Bosyk},
  \citenamefont {Bellomo}, \citenamefont {Holik}, \citenamefont {Freytes},\
  and\ \citenamefont
  {Sergioli}}]{BosykBellomoHolikFreytes_OptimalCommonResource_2019}%
  \BibitemOpen
  \bibfield  {author} {\bibinfo {author} {\bibfnamefont {G.~M.}\ \bibnamefont
  {Bosyk}}, \bibinfo {author} {\bibfnamefont {G.}~\bibnamefont {Bellomo}},
  \bibinfo {author} {\bibfnamefont {F.}~\bibnamefont {Holik}}, \bibinfo
  {author} {\bibfnamefont {H.}~\bibnamefont {Freytes}},\ and\ \bibinfo {author}
  {\bibfnamefont {G.}~\bibnamefont {Sergioli}},\ }\bibfield  {title} {\bibinfo
  {title} {Optimal common resource in majorization-based resource theories},\
  }\href {https://doi.org/10.1088/1367-2630/ab3734} {\bibfield  {journal}
  {\bibinfo  {journal} {New J. Phys.}\ }\textbf {\bibinfo {volume} {21}},\
  \bibinfo {pages} {083028} (\bibinfo {year} {2019})}\BibitemShut {NoStop}%
\bibitem [{\citenamefont {Massri}\ \emph {et~al.}(2020)\citenamefont {Massri},
  \citenamefont {Bellomo}, \citenamefont {Holik},\ and\ \citenamefont
  {Bosyk}}]{MassriEtAl_ExtremalElementsOfASublattice_2020}%
  \BibitemOpen
  \bibfield  {author} {\bibinfo {author} {\bibfnamefont {C.}~\bibnamefont
  {Massri}}, \bibinfo {author} {\bibfnamefont {G.}~\bibnamefont {Bellomo}},
  \bibinfo {author} {\bibfnamefont {F.}~\bibnamefont {Holik}},\ and\ \bibinfo
  {author} {\bibfnamefont {G.~M.}\ \bibnamefont {Bosyk}},\ }\bibfield  {title}
  {\bibinfo {title} {Extremal elements of a sublattice of the majorization
  lattice and approximate majorization},\ }\href
  {https://doi.org/10.1088/1751-8121/ab8674} {\bibfield  {journal} {\bibinfo
  {journal} {J. Phys. A: Math. Theor.}\ }\textbf {\bibinfo {volume} {53}},\
  \bibinfo {pages} {215305} (\bibinfo {year} {2020})}\BibitemShut {NoStop}%
\bibitem [{\citenamefont {Sauerwein}\ \emph {et~al.}(2018)\citenamefont
  {Sauerwein}, \citenamefont {Schwaiger},\ and\ \citenamefont
  {Kraus}}]{SauerweinSchwaigerKraus_DiscreteAndDifferentiableEntanglementTransformations_2018}%
  \BibitemOpen
  \bibfield  {author} {\bibinfo {author} {\bibfnamefont {D.}~\bibnamefont
  {Sauerwein}}, \bibinfo {author} {\bibfnamefont {K.}~\bibnamefont
  {Schwaiger}},\ and\ \bibinfo {author} {\bibfnamefont {B.}~\bibnamefont
  {Kraus}},\ }\bibfield  {title} {\bibinfo {title} {Discrete {And}
  {Differentiable} {Entanglement} {Transformations}},\ }\href
  {https://doi.org/10.48550/arXiv.1808.02819} {\bibfield  {journal} {\bibinfo
  {journal} {arXiv:1808.02819}\ } (\bibinfo {year} {2018})}\BibitemShut
  {NoStop}%
\bibitem [{\citenamefont {de~Oliveira~Junior}\ \emph
  {et~al.}(2022)\citenamefont {de~Oliveira~Junior}, \citenamefont {Czartowski},
  \citenamefont {Życzkowski},\ and\ \citenamefont
  {Korzekwa}}]{deOliveiraJuniorEtAl_GeometricStructure_2022}%
  \BibitemOpen
  \bibfield  {author} {\bibinfo {author} {\bibfnamefont {A.}~\bibnamefont
  {de~Oliveira~Junior}}, \bibinfo {author} {\bibfnamefont {J.}~\bibnamefont
  {Czartowski}}, \bibinfo {author} {\bibfnamefont {K.}~\bibnamefont
  {Życzkowski}},\ and\ \bibinfo {author} {\bibfnamefont {K.}~\bibnamefont
  {Korzekwa}},\ }\bibfield  {title} {\bibinfo {title} {Geometric structure of
  thermal cones},\ }\href {https://doi.org/10.1103/PhysRevE.106.064109}
  {\bibfield  {journal} {\bibinfo  {journal} {Phys. Rev. E}\ }\textbf {\bibinfo
  {volume} {106}},\ \bibinfo {pages} {064109} (\bibinfo {year}
  {2022})}\BibitemShut {NoStop}%
\bibitem [{\citenamefont {Yu}\ and\ \citenamefont
  {Gühne}(2019)}]{YuGühne_DetectingCoherence_2019}%
  \BibitemOpen
  \bibfield  {author} {\bibinfo {author} {\bibfnamefont {X.-D.}\ \bibnamefont
  {Yu}}\ and\ \bibinfo {author} {\bibfnamefont {O.}~\bibnamefont {Gühne}},\
  }\bibfield  {title} {\bibinfo {title} {Detecting coherence via spectrum
  estimation},\ }\href {https://doi.org/10.1103/PhysRevA.99.062310} {\bibfield
  {journal} {\bibinfo  {journal} {Phys. Rev. A}\ }\textbf {\bibinfo {volume}
  {99}},\ \bibinfo {pages} {062310} (\bibinfo {year} {2019})}\BibitemShut
  {NoStop}%
\bibitem [{\citenamefont {Guo}\ \emph {et~al.}(2016)\citenamefont {Guo},
  \citenamefont {Chitambar},\ and\ \citenamefont
  {Duan}}]{GuoChitambarDuan_CommonResourceState_2016}%
  \BibitemOpen
  \bibfield  {author} {\bibinfo {author} {\bibfnamefont {C.}~\bibnamefont
  {Guo}}, \bibinfo {author} {\bibfnamefont {E.}~\bibnamefont {Chitambar}},\
  and\ \bibinfo {author} {\bibfnamefont {R.}~\bibnamefont {Duan}},\ }\bibfield
  {title} {\bibinfo {title} {Common {Resource} {State} for {Preparing}
  {Multipartite} {Quantum} {Systems} via {Local} {Operations} and {Classical}
  {Communication}},\ }\href {https://doi.org/10.48550/arXiv.1601.06220}
  {\bibfield  {journal} {\bibinfo  {journal} {arXiv:1601.06220}\ } (\bibinfo
  {year} {2016})}\BibitemShut {NoStop}%
\bibitem [{\citenamefont {Vidal}\ \emph {et~al.}(2000)\citenamefont {Vidal},
  \citenamefont {Jonathan},\ and\ \citenamefont
  {Nielsen}}]{VidalJonathanNielsen_ApproximateTransformations_2000}%
  \BibitemOpen
  \bibfield  {author} {\bibinfo {author} {\bibfnamefont {G.}~\bibnamefont
  {Vidal}}, \bibinfo {author} {\bibfnamefont {D.}~\bibnamefont {Jonathan}},\
  and\ \bibinfo {author} {\bibfnamefont {M.~A.}\ \bibnamefont {Nielsen}},\
  }\bibfield  {title} {\bibinfo {title} {Approximate transformations and robust
  manipulation of bipartite pure-state entanglement},\ }\href
  {https://doi.org/10.1103/PhysRevA.62.012304} {\bibfield  {journal} {\bibinfo
  {journal} {Phys. Rev. A}\ }\textbf {\bibinfo {volume} {62}},\ \bibinfo
  {pages} {012304} (\bibinfo {year} {2000})}\BibitemShut {NoStop}%
\bibitem [{\citenamefont {Jonathan}\ and\ \citenamefont
  {Plenio}(1999{\natexlab{a}})}]{JonathanPlenio_MinimalConditions_1999}%
  \BibitemOpen
  \bibfield  {author} {\bibinfo {author} {\bibfnamefont {D.}~\bibnamefont
  {Jonathan}}\ and\ \bibinfo {author} {\bibfnamefont {M.~B.}\ \bibnamefont
  {Plenio}},\ }\bibfield  {title} {\bibinfo {title} {Minimal conditions for
  local pure-state entanglement manipulation},\ }\href
  {https://doi.org/10.1103/PhysRevLett.83.1455} {\bibfield  {journal} {\bibinfo
   {journal} {Phys. Rev. Lett.}\ }\textbf {\bibinfo {volume} {83}},\ \bibinfo
  {pages} {1455} (\bibinfo {year} {1999}{\natexlab{a}})}\BibitemShut {NoStop}%
\bibitem [{\citenamefont {Jonathan}\ and\ \citenamefont
  {Plenio}(1999{\natexlab{b}})}]{JonathanPlenio_Catalyst_1999}%
  \BibitemOpen
  \bibfield  {author} {\bibinfo {author} {\bibfnamefont {D.}~\bibnamefont
  {Jonathan}}\ and\ \bibinfo {author} {\bibfnamefont {M.~B.}\ \bibnamefont
  {Plenio}},\ }\bibfield  {title} {\bibinfo {title} {{Entanglement-Assisted
  Local Manipulation of Pure Quantum States}},\ }\href
  {https://doi.org/10.1103/PhysRevLett.83.3566} {\bibfield  {journal} {\bibinfo
   {journal} {Phys. Rev. Lett.}\ }\textbf {\bibinfo {volume} {83}},\ \bibinfo
  {pages} {3566} (\bibinfo {year} {1999}{\natexlab{b}})}\BibitemShut {NoStop}%
\bibitem [{\citenamefont {Perseguers}\ \emph {et~al.}(2010)\citenamefont
  {Perseguers}, \citenamefont {Lewenstein}, \citenamefont {Acín},\ and\
  \citenamefont {Cirac}}]{Perseguers_QuantumRandomNetworks_2010}%
  \BibitemOpen
  \bibfield  {author} {\bibinfo {author} {\bibfnamefont {S.}~\bibnamefont
  {Perseguers}}, \bibinfo {author} {\bibfnamefont {M.}~\bibnamefont
  {Lewenstein}}, \bibinfo {author} {\bibfnamefont {A.}~\bibnamefont {Acín}},\
  and\ \bibinfo {author} {\bibfnamefont {J.~I.}\ \bibnamefont {Cirac}},\
  }\bibfield  {title} {\bibinfo {title} {Quantum random networks},\ }\href
  {https://doi.org/https://doi.org/10.1038/nphys1665} {\bibfield  {journal}
  {\bibinfo  {journal} {Nat. Phys.}\ }\textbf {\bibinfo {volume} {6}},\
  \bibinfo {pages} {539} (\bibinfo {year} {2010})}\BibitemShut {NoStop}%
\bibitem [{\citenamefont {Malik}\ \emph {et~al.}(2022)\citenamefont {Malik},
  \citenamefont {Meng}, \citenamefont {Havlin}, \citenamefont {Korniss},
  \citenamefont {Szymanski},\ and\ \citenamefont
  {Gao}}]{Malik_ConcurrencePercolation_2022}%
  \BibitemOpen
  \bibfield  {author} {\bibinfo {author} {\bibfnamefont {O.}~\bibnamefont
  {Malik}}, \bibinfo {author} {\bibfnamefont {X.}~\bibnamefont {Meng}},
  \bibinfo {author} {\bibfnamefont {S.}~\bibnamefont {Havlin}}, \bibinfo
  {author} {\bibfnamefont {G.}~\bibnamefont {Korniss}}, \bibinfo {author}
  {\bibfnamefont {B.~K.}\ \bibnamefont {Szymanski}},\ and\ \bibinfo {author}
  {\bibfnamefont {J.}~\bibnamefont {Gao}},\ }\bibfield  {title} {\bibinfo
  {title} {Concurrence percolation threshold of large-scale quantum networks},\
  }\href {https://doi.org/https://doi.org/10.1038/s42005-022-00958-4}
  {\bibfield  {journal} {\bibinfo  {journal} {Commun. Phys.}\ }\textbf
  {\bibinfo {volume} {5}},\ \bibinfo {pages} {193} (\bibinfo {year}
  {2022})}\BibitemShut {NoStop}%
\bibitem [{\citenamefont {Horodecki}\ and\ \citenamefont
  {Oppenheim}(2013)}]{HorodeckiOppenheim_FundamentalLimitations_2013}%
  \BibitemOpen
  \bibfield  {author} {\bibinfo {author} {\bibfnamefont {M.}~\bibnamefont
  {Horodecki}}\ and\ \bibinfo {author} {\bibfnamefont {J.}~\bibnamefont
  {Oppenheim}},\ }\bibfield  {title} {\bibinfo {title} {Fundamental limitations
  for quantum and nanoscale thermodynamics},\ }\href
  {https://doi.org/https://doi.org/10.1038/ncomms3059} {\bibfield  {journal}
  {\bibinfo  {journal} {Nat. Commun.}\ }\textbf {\bibinfo {volume} {4}},\
  \bibinfo {pages} {2059} (\bibinfo {year} {2013})}\BibitemShut {NoStop}%
\bibitem [{\citenamefont {Brandão}\ \emph {et~al.}(2015)\citenamefont
  {Brandão}, \citenamefont {Horodecki}, \citenamefont {Ng}, \citenamefont
  {Oppenheim},\ and\ \citenamefont {Wehner}}]{Brandao_TheSecondLaws_2015}%
  \BibitemOpen
  \bibfield  {author} {\bibinfo {author} {\bibfnamefont {F.}~\bibnamefont
  {Brandão}}, \bibinfo {author} {\bibfnamefont {M.}~\bibnamefont {Horodecki}},
  \bibinfo {author} {\bibfnamefont {N.}~\bibnamefont {Ng}}, \bibinfo {author}
  {\bibfnamefont {J.}~\bibnamefont {Oppenheim}},\ and\ \bibinfo {author}
  {\bibfnamefont {S.}~\bibnamefont {Wehner}},\ }\bibfield  {title} {\bibinfo
  {title} {The second laws of quantum thermodynamics},\ }\href
  {https://doi.org/https://doi.org/10.1073/pnas.1411728112} {\bibfield
  {journal} {\bibinfo  {journal} {Proc. Natl. Acad. Sci.}\ }\textbf {\bibinfo
  {volume} {112}},\ \bibinfo {pages} {3275} (\bibinfo {year}
  {2015})}\BibitemShut {NoStop}%
\bibitem [{\citenamefont {Singh}\ \emph {et~al.}(2021)\citenamefont {Singh},
  \citenamefont {Das},\ and\ \citenamefont {Cerf}}]{Singh_PartialOrder_2021}%
  \BibitemOpen
  \bibfield  {author} {\bibinfo {author} {\bibfnamefont {U.}~\bibnamefont
  {Singh}}, \bibinfo {author} {\bibfnamefont {S.}~\bibnamefont {Das}},\ and\
  \bibinfo {author} {\bibfnamefont {N.~J.}\ \bibnamefont {Cerf}},\ }\bibfield
  {title} {\bibinfo {title} {Partial order on passive states and {Hoffman}
  majorization in quantum thermodynamics},\ }\href
  {https://doi.org/10.1103/PhysRevResearch.3.033091} {\bibfield  {journal}
  {\bibinfo  {journal} {Phys. Rev. Res.}\ }\textbf {\bibinfo {volume} {3}},\
  \bibinfo {pages} {033091} (\bibinfo {year} {2021})}\BibitemShut {NoStop}%
\bibitem [{\citenamefont {Neven}\ \emph {et~al.}(2021)\citenamefont {Neven},
  \citenamefont {Gunn}, \citenamefont {Hebenstreit},\ and\ \citenamefont
  {Kraus}}]{Neven_LocalTransformations_2021}%
  \BibitemOpen
  \bibfield  {author} {\bibinfo {author} {\bibfnamefont {A.}~\bibnamefont
  {Neven}}, \bibinfo {author} {\bibfnamefont {D.}~\bibnamefont {Gunn}},
  \bibinfo {author} {\bibfnamefont {M.}~\bibnamefont {Hebenstreit}},\ and\
  \bibinfo {author} {\bibfnamefont {B.}~\bibnamefont {Kraus}},\ }\bibfield
  {title} {\bibinfo {title} {Local transformations of multiple multipartite
  states},\ }\href {https://doi.org/doi: 10.21468/SciPostPhys.11.2.042}
  {\bibfield  {journal} {\bibinfo  {journal} {SciPost Phys.}\ }\textbf
  {\bibinfo {volume} {11}},\ \bibinfo {pages} {042} (\bibinfo {year}
  {2021})}\BibitemShut {NoStop}%
\bibitem [{\citenamefont {Chitambar}\ \emph {et~al.}(2008)\citenamefont
  {Chitambar}, \citenamefont {Duan},\ and\ \citenamefont
  {Shi}}]{ChitambarDuanShi_Tripartite_2008}%
  \BibitemOpen
  \bibfield  {author} {\bibinfo {author} {\bibfnamefont {E.}~\bibnamefont
  {Chitambar}}, \bibinfo {author} {\bibfnamefont {R.}~\bibnamefont {Duan}},\
  and\ \bibinfo {author} {\bibfnamefont {Y.}~\bibnamefont {Shi}},\ }\bibfield
  {title} {\bibinfo {title} {{Tripartite Entanglement Transformations and
  Tensor Rank}},\ }\href {https://doi.org/10.1103/PhysRevLett.101.140502}
  {\bibfield  {journal} {\bibinfo  {journal} {Phys. Rev. Lett.}\ }\textbf
  {\bibinfo {volume} {101}},\ \bibinfo {pages} {140502} (\bibinfo {year}
  {2008})}\BibitemShut {NoStop}%
\bibitem [{\citenamefont {Dür}\ \emph {et~al.}(2000)\citenamefont {Dür},
  \citenamefont {Vidal},\ and\ \citenamefont
  {Cirac}}]{DurVidalCirac_ThreeQubitsCanBeEntangled_2000}%
  \BibitemOpen
  \bibfield  {author} {\bibinfo {author} {\bibfnamefont {W.}~\bibnamefont
  {Dür}}, \bibinfo {author} {\bibfnamefont {G.}~\bibnamefont {Vidal}},\ and\
  \bibinfo {author} {\bibfnamefont {J.~I.}\ \bibnamefont {Cirac}},\ }\bibfield
  {title} {\bibinfo {title} {Three qubits can be entangled in two inequivalent
  ways},\ }\href {https://doi.org/10.1103/PhysRevA.62.062314} {\bibfield
  {journal} {\bibinfo  {journal} {Phys. Rev. A}\ }\textbf {\bibinfo {volume}
  {62}},\ \bibinfo {pages} {062314} (\bibinfo {year} {2000})}\BibitemShut
  {NoStop}%
\end{thebibliography}%

\end{document}